\documentclass[sigconf]{acmart}

\usepackage[linesnumbered,ruled,vlined]{algorithm2e}
\usepackage{multirow}
\usepackage{enumitem}
\usepackage{balance}
\newtheorem{problem}{Problem}

\newcommand{\roundA}{}
\newcommand{\roundB}{}
\newcommand{\roundC}{}

\newcommand{\roundD}{}
\newcommand{\roundE}{}
\newcommand{\roundF}{}
\newcommand{\roundG}{}
\newcommand{\ChengComment}{}
\newcommand{\ChengCommentB}{}
\newcommand{\ChengCommentC}{}
\newcommand{\Yuchange}{}
\newcommand{\revision}{}

\AtBeginDocument{%
  \providecommand\BibTeX{{%
    \normalfont B\kern-0.5em{\scshape i\kern-0.25em b}\kern-0.8em\TeX}}}

\setcopyright{acmcopyright}
\copyrightyear{2022}
\acmYear{2022}
\acmDOI{XXXX/XXXXX}

\acmConference[SIGMOD '22]{SIGMOD '22: Proceedings of the 2022 International Conference on Management of Data}{June 12--17, 2022}{Philadelphia, PA}
\acmBooktitle{SIGMOD '22: Proceedings of the 2022 International Conference on Management of Data, June 12--17, 2022, Philadelphia, PA}
\acmPrice{15.00}
\acmISBN{978-1-4503-XXXX-X/18/06}

\begin{document}
\sloppy
\title{Efficient Algorithms for Maximal $k$-Biplex Enumeration}

\author{Kaiqiang Yu}
\affiliation{%
  \institution{Nanyang Technological University}
}
\email{kaiqiang002@e.ntu.edu.sg}
\author{Cheng Long}
\authornotemark[1]
\thanks{*Cheng Long is the corresponding author.}
\affiliation{%
  \institution{Nanyang Technological University}
}
\email{c.long@ntu.edu.sg}
\author{Shengxin Liu}
\affiliation{%
   \institution{Harbin Institute of Technology, Shenzhen}
}
\email{sxliu@hit.edu.cn}
\author{Da Yan}
\affiliation{%
  \institution{University of Alabama at Birmingham}
}
\email{yanda@uab.edu}



\begin{abstract}
Mining maximal subgraphs with cohesive structures from a bipartite graph {\ChengComment has been widely studied}. One important cohesive structure on bipartite graphs is {\roundF $k$-biplex}, where each vertex on one side disconnects at most $k$ vertices on the other side. In this paper, we study the \textit{maximal} $k$\textit{-biplex enumeration} problem which {\roundF enumerates} all maximal $k$-biplexes. 
Existing methods {\ChengComment suffer from efficiency and/or scalability issues and have the time of waiting for the next output exponential {\roundF w.r.t.} the size of the input bipartite graph (i.e., an exponential delay).}
{\roundA In this paper, we adopt a reverse search framework called \texttt{bTraversal},
which corresponds to a depth-first search (DFS) {\roundF procedure} on an implicit {\roundF solution graph} on top of all maximal $k$-biplexes. 
We then develop a series of techniques for improving and implementing {\roundF this} framework including (1) carefully selecting an initial solution to start DFS, (2) pruning the vast majority of links from the solution graph of \texttt{bTraversal}, and (3) implementing abstract procedures of the framework. 
The resulting algorithm is called \texttt{iTraversal}, which has its underlying solution graph significantly sparser than {\roundF (around 0.1\% of)} that of \texttt{bTraversal}. 
Besides, \texttt{iTraversal} provides {\roundF a guarantee of} polynomial delay. 
}
{\roundF Our} experimental results on real and synthetic graphs, where the largest one contains one billion edges, show that our algorithm is up to four orders of magnitude faster than existing algorithms.

\end{abstract}





\maketitle

\section{Introduction}
\label{sec:introduction}
In many applications, two types of entities are involved and interact with each other. Some examples include (1) {\roundF social media} where users comment on articles~\cite{zhang2020overview}, (2) e-commerce services where customers post reviews on or purchase products~\cite{wang2006unifying}, (3) {\roundF collaboration networks where authors publish papers~\cite{ley2002dblp}}, etc. In these applications, the two types of entities and the interactions between them can be naturally modelled as a \emph{bipartite graph} {\ChengComment with the entities being vertices and the interactions being edges.}

{\roundD
For a given bipartite graph, a \emph{dense}/\emph{cohesive} subgraph within it usually carries interesting information that can be used for solving practical problems such as fraud detection~\cite{yu2021graph,gangireddy2020unsupervised}, online recommendation~\cite{DBLP:conf/kdd/PoernomoG09a, DBLP:conf/icdm/GunnemannMRS11} and community search~\cite{DBLP:conf/ssdbm/HaoZW020,DBLP:journals/corr/abs-2011-08399}. }
{\ChengComment For example, 
in social networking applications, when a group of users are paid to promote a specific set of fake articles via retweets, the induced subgraph by these users and the articles would be dense. Identifying such dense subgraphs would help to detect the fake users and articles~\cite{gangireddy2020unsupervised}.
%
{\ChengCommentC As a second example, in e-commerce services, after identifying a dense subgraph between customers and products, it would be natural to recommend {products to} those customers which disconnect the products within the subgraph
~\cite{amer2009group}.
}
}

Quite a few definitions have been proposed for a dense bipartite graph, including biclique~\cite{zhang2014finding},
$(\alpha,\beta)$-core~\cite{liu2019efficient}, 
$k$-bitruss~\cite{wang2020efficient},
$k$-biplex~\cite{DBLP:journals/sadm/SimLGL09,yu2021efficient},
$\delta$-quasi-biclique~\cite{DBLP:conf/cocoon/LiuLW08}, etc. 
These definitions impose different requirements on the connections and/or disconnections within a subgraph. For example, biclique requires that any vertex from one side \emph{connects} \emph{all} vertices from the other side, $(\alpha, \beta)$-core requires that any vertex from one side \emph{connects} at least \emph{a certain} number of vertices from the other side, and $k$-biplex requires that any vertex from one side \emph{disconnects} \emph{at most} $k$ vertices from the other side, where $k$ is a small positive integer {\roundF specified by users}.

In this paper, we study the problem of enumerating \emph{\underline{M}aximal $k$-\underline{B}i\underline{P}lexes}, called \emph{MBP}s, for the following considerations.
First, $k$-biplex imposes a strict enough requirement on the connections within a subgraph yet allows some disconnections which are common in real applications (due to data quality issues such as {\roundF incomplete data}). 
%
Second, $k$-biplex satisfies the \emph{hereditary property}~\cite{DBLP:journals/jcss/CohenKS08}, i.e., any subgraph of a $k$-biplex is {\roundF still a $k$-biplex}. This can be utilized to design efficient frameworks for {\roundF enumerating} MBPs (details will be discussed later).
%
Third, other definitions have some shortcomings when used for the aforementioned applications. Specifically, 
(1)~biclique may impose a too strict requirement on the connections, i.e., with even one single connection missed from a biclique, the subgraph is no longer a biclique. (2)~$(\alpha, \beta)$-core and $k$-bitruss can be computed efficiently, but they 
impose no constraints on the disconnections within a subgraph, i.e., a vertex may disconnect many vertices from the other side. (3)~$\delta$-quasi-biclique does not satisfy the hereditary property {\roundF so} enumerating the corresponding maximal structures is much harder than enumerating MBPs. 
{\ChengComment We {\roundF have conducted} some case studies, which show that $k$-biplex works better for fraud detection on e-commerce platforms and captures more cohesive subgraphs than other definitions including biclique and $(\alpha, \beta)$-core (details will be presented in Section~\ref{subsec:case-study}).}

There are several existing methods, which can be used or adapted for enumerating MBPs. The first one is called \texttt{iMB}~\cite{DBLP:journals/sadm/SimLGL09,yu2021efficient}{\roundF, which} uses two prefix trees {\roundF to organize} the subsets of vertices of {\roundF the two} sides and searches MBPs with backtracking and various pruning strategies. Nevertheless, it suffers from {\roundG two} issues: (1) it is not scalable and cannot handle big graphs since most of its pruning techniques {\roundF rely highly} on some size constraints imposed on the $k$-biplexes to enumerate; (2) the delay, which represents the amount of {\roundF waiting time} for the next MBP {\roundF to return or for} the termination of the algorithm, is exponential w.r.t. the number of vertices of the bipartite graph. The second {\roundF baseline} is based on \emph{graph inflation}. Specifically, it first inflates a given bipartite graph into a general one by including an edge between every pair of two vertices at the same side and then enumerates all maximal $(k\!+\!1)$-plexes on the inflated general graphs (for which, the state-of-the-art is \texttt{FaPlexen}~\cite{DBLP:conf/aaai/ZhouXGXJ20}). Here, a $(k+1)$-plex on a general graph represents a subgraph, where each vertex $v$ disconnects at most $(k+1)$ vertices (including $v$) within the subgraph~\cite{DBLP:conf/sigmod/BerlowitzCK15}. The correctness of this method is based on the fact that a $k$-biplex on the bipartite graph corresponds to a $(k+1)$-plex on the inflated general graph. Nevertheless, the graph inflation step would usually generate very dense graphs 
and enumerating $(k+1)$-plexes on a dense graph is rather time-consuming. 

{\roundA In this paper, we adopt a \emph{reverse search} framework~\cite{DBLP:journals/jcss/CohenKS08} which we call \texttt{bTraversal} for enumerating MBPs.  \texttt{bTraversal} is originally designed for enumerating maximal subgraph structures that satisfy the hereditary property (each such structure is called a \emph{solution}). 
%
The key insight is that given a solution $H$, it is possible to find another solution by \emph{excluding} some existing vertices from $H$ and then \emph{including} some new ones to $H$. %
Specifically, \texttt{bTraversal} involves two steps. First, it finds one solution $H_0$ as the initial one, which can be any one among all solutions. 
Second, it finds solutions from $H_0$ via a procedure of excluding and including vertices from and to $H_0$ and then recursively performs the procedure from those found solutions until no new solutions are found. Suppose we take each solution as a vertex and create a directed edge from solution $H$ to solution $H'$ if \texttt{bTraversal} can find $H'$ from $H$ via the aforementioned procedure of excluding and including vertices. Then, we obtain a graph structure on top of all {\roundF found} solutions, which is called a \emph{solution graph}~\cite{DBLP:conf/stoc/ConteU19}. We refer to the vertices and directed edges in the solution graph as nodes and links, respectively. Then, \texttt{bTraversal} corresponds to a \emph{depth-first search} (DFS) procedure over the solution graph. According to~\cite{DBLP:conf/stoc/ConteU19}, the solution graph constructed in this way is \emph{strongly connected} and thus \texttt{bTraversal} is able to enumerate all solutions starting from any solution.
}

Nevertheless, \texttt{bTraversal} is still insufficient in the following aspects. First, \texttt{bTraversal} requires that any solution should be reachable from any other solution in the solution graph (i.e., the solution graph is strongly connected) so that it can find \emph{all} solutions from \emph{any} initial solution. To fulfill this requirement, it would find many solutions from one solution. Consequently, the underlying solution graph tends to be dense and the DFS procedure on the solution graph would be costly. Note that the time complexity of DFS is proportional to the number of links in the solution graph. 
Second, \texttt{bTraversal} is originally designed for general structures that satisfy the hereditary property, but not just for the $k$-biplexes. Consequently, it overlooks those unique characteristics of $k$-biplexes that would otherwise help to improve the algorithm.
Third, \texttt{bTraversal} does not support enumerating MBPs with size at least a threshold (called large MBPs) - it {\roundF has} to enumerate all MBPs first and then filter out the MBPs violating the size constraint, which is inefficient.

We observe that the requirement of a strongly connected solution graph by \texttt{bTraversal} is stronger than necessary. In fact, it would be sufficient as long as all solutions are reachable from some \emph{specific} solution since we can then start the DFS procedure from this solution {\roundF to} reach all solutions. Motivated by this observation, in this paper, we propose an improved framework called \texttt{iTraversal}, which {\roundF begins} DFS from a carefully selected solution. Specifically, it selects $H_0 = (L_0, \mathcal{R})$ as the initial solution, where $\mathcal{R}$ is the set containing all vertices from the right side of the bipartite graph and $L_0$ is any maximal set of vertices from the left side with $(L_0, \mathcal{R})$ being a $k$-biplex. With this designated initial solution, \texttt{iTraversal} makes it possible to significantly sparsify the solution graph that is defined by \texttt{bTraversal} while maintaining that all solutions are reachable from this initial solution. Specifically, we develop a series of three techniques for sparsifying the solution graph, namely (1) left-anchored traversal, (2) right-shrinking traversal, and (3) exclusion strategy. Technique (1) is based on a discovery that in one step of \emph{including} vertices {\roundF to generate new} solutions from a certain solution, pruning all vertices on the right side from being included would still guarantee that all solutions are reachable from {\roundF our} initial solution. Technique (2) is based on another discovery that by retaining only those links from a solution $H= (L, R)$ to another solution $H' = (L', R')$ with $R'\subseteq R$ and removing all {\roundF the} other links, all solutions are still reachable from the initial solution. Technique (3) is a technique that prunes a vertex from being included {\roundF to find new} solutions from a solution if the vertex appears in an \emph{exclusion set} that is maintained during the running {\roundF of \texttt{iTraversal}}. 
Based on our experimental results, the number of links in the solution graph of \texttt{iTraversal} sparsified with the three techniques is about 0.1\% of that in the original solution graph of \texttt{bTraversal}.

{\roundA In summary, our major contributions are {\roundF summarized} as follows.
\begin{itemize}[leftmargin=*]
    \item We propose a new framework \texttt{iTraversal} for enumerating MBPs. We further develop three techniques, namely left-anchored traversal, right-shrinking traversal, and exclusion strategy, for sparsifying the solution graph under \texttt{iTraversal}. We prove that the delay of finding the next solution with \texttt{iTraversal} is \emph{polynomial} w.r.t. the number of vertices and improves that of the conventional \texttt{bTraversal} {\roundF framework} (Section~\ref{sec:Traversal}). 
    {\ChengCommentC We remark that (1) the first two techniques are novel and work only with \texttt{iTraversal} but not with \texttt{bTraversal} and (2) the third technique  was proposed for \texttt{bTraversal} but its correctness for \texttt{iTraversal}, which is not trivial to prove, is verified in this paper.}
    
    \item We design {\roundF an efficient algorithm for a key procedure} that is involved in the \texttt{iTraversal} framework, which is to enumerate solutions within a graph that \emph{almost} satisfies the definition of $k$-biplex 
    (Section~\ref{sec:enumalmostsat}).
    
    \item {\ChengComment We extend \texttt{iTraversal} to enumerate those MBPs with the size of at least a threshold (i.e., large MBPs) without enumerating all MBPs, which is not possible {\roundF when using} the conventional \texttt{bTraversal}} {\roundF framework} (Section~\ref{sec:size-constrained}).
    
    \item We conduct extensive experiments on both real and synthetic datasets, which verify that {\ChengCommentB (1) $k$-biplex works better in a fraud detection task than some other structures including biclique and $(\alpha, \beta)$-core and captures cohesive subgraphs and} {\revision (2) the proposed algorithms with new techniques are up to {\roundB four} orders of magnitude faster than existing algorithms including the one based on \texttt{bTaversal} (Section~\ref{sec:exp}).}
\end{itemize}  }
{\roundB {\roundF Among other sections, we define our} problem in Section~\ref{sec:problem}, review the related work in Section~\ref{related} and conclude the paper in Section~\ref{sec:conclusion}.}

\section{Problem Definition}
\label{sec:problem}

In this paper, we consider an undirected and unweighted bipartite graph $G = (\mathcal{L} \cup \mathcal{R},\mathcal{E})$ with two disjoint vertex sets $\mathcal{L}$, $\mathcal{R}$ and an edge set $\mathcal{E}$. $\mathcal{L}$ and $\mathcal{R}$ are supposed to be on the left and right side, respectively.
{\roundA We denote by $V(G)$ the set of vertices in $G$, i.e., $V(G) = \mathcal{L} \cup \mathcal{R}$, and by $E(G)$ the set of edges in $G$, i.e., $E(G) = \mathcal{E}$.}
{\roundA Given  $L \subseteq \mathcal{L}$ and $R \subseteq \mathcal{R}$, the {\em induced (bipartite) subgraph} $G[L \cup R]$ of $G$ consists of the set of vertices $L\cup R$ and the set of edges between $L$ and $R$.}
{
Note that all subgraphs mentioned in this paper refer to induced subgraphs, and we use $H$ or $(L,R)$ as a shorthand of $H=G[L\cup R]$.}  

{\roundC
Given $v\in \mathcal{L}$ and $R\subseteq \mathcal{R}$,
we define $\Gamma(v, R)$ (resp. $\overline{\Gamma}(v, R)$) to be the set of vertices that are in $R$ and connect (resp. disconnect) $v$, i.e., $\Gamma(v, R)=\{u\mid (v,u)\in \mathcal{E} \text{ and } u \in R \}$ (resp. $\overline {\Gamma}(v, R)=\{u\mid (v,u)\not \in \mathcal{E} $ and $u \in R \}$). 
Note that $\Gamma(v, R) \cup \overline{\Gamma}(v, R) = R$.
In addition, we define $\delta(v,R) = |\Gamma(v, R)|$ and $\overline {\delta}(v,R) = |\overline {\Gamma}(v, R)|$.
Given $u\in \mathcal{R}$ and $L\subseteq \mathcal{L}$,
$\Gamma(u,L)$ (resp. $\overline {\Gamma}(u,L)$) and $\delta(u,L)$ (resp. $\overline {\delta}(u,L)$) are similarly defined. Next, we introduce the cohesive bipartite structure $k$-biplex that is exploited in this paper.
}
\begin{definition}[$k$-biplex \cite{DBLP:journals/sadm/SimLGL09}]
	{\revision Let $k$ be a positive integer.} An induced subgraph $G[L \cup R]$ of a bipartite graph $G$ is said to be a {\em $k$-biplex} if {\roundB (1) $\overline {\delta}(v,R) \leq k$, $\forall v \in L$ and (2) $\overline {\delta}(u,L) \leq k$, $\forall u \in R$.}
	\label{definition:k-biplex}
\end{definition}
For a $k$-biplex, parameter $k$ represents the number of missing edges that each vertex in $G[L\cup R]$ can tolerate.
{\ChengCommentB We note that it is possible to use different $k$'s at different sides and the techniques developed in this paper can be easily adapted to this case.}
Moreover, the $k$-biplex structures satisfy the \emph{hereditary property}~\cite{DBLP:journals/jcss/CohenKS08}, which we present in the following lemma.
\begin{lemma} [Hereditary Property]
	If $H=(L,R)$ is a $k$-biplex, any subgraph $H'=(L',R')$ of $H$ with $L'\subseteq L, R'\subseteq R$ is a $k$-biplex.
\end{lemma}
This can be easily verified by the fact that with some vertices excluded from a $k$-biplex, each remaining vertex has the number of disconnections non-increasing, i.e., still bounded by $k$. 

As there might exist a large number of $k$-biplexes, one common practice is to return a compact representation of the set of all $k$-biplexes, namely the set of all \emph{maximal $k$-biplexes}.

\begin{definition}[Maximal $k$-biplex]
	A $k$-biplex $G[L\cup R]$ is said to be \emph{maximal} if and only if there is no other $k$-biplex $G[L'\cup R']$ which is a superset of $G[L\cup R]$ (i.e., $L \subseteq L'$ and $R \subseteq R'$).
	\label{definition:maximal-k-biplex}
\end{definition}
In this paper, we use MBP as a shorthand of maximal $k$-biplex when the context is clear.
We are ready to formalize the problem exploited in this paper:
\begin{problem}[Maximal $k$-biplex Enumeration \cite{DBLP:journals/sadm/SimLGL09}]
	Given a bipartite graph $G = (\mathcal{L} \cup \mathcal{R}, \mathcal{E})$ and a positive integer $k$, 
	the {\em Maximal $k$-biplex Enumeration Problem} aims to report all MBPs.
\end{problem}

We call each maximal induced subgraph of graph $G$ that is a $k$-biplex as a \emph{solution}. The maximal $k$-biplex enumeration problem is to enumerate all solutions. We use the terms ``maximal $k$-biplex (MBP)'' and ``solution'' interchangeably throughout this paper.

\section{The \texttt{iTraversal} Algorithm}
\label{sec:Traversal}

\begin{figure}[]
	\centering
	\includegraphics[width=0.77\linewidth]{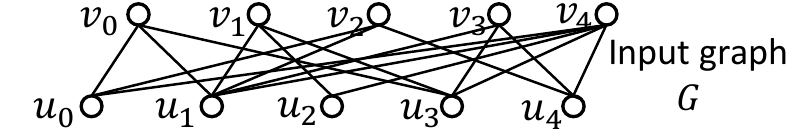}
	\vspace{-0.15in}
	\caption{Input graph used throughout the paper.}
	\label{fig:input_graph}
	\vspace{0.12in}
	\includegraphics[width=0.77\linewidth]{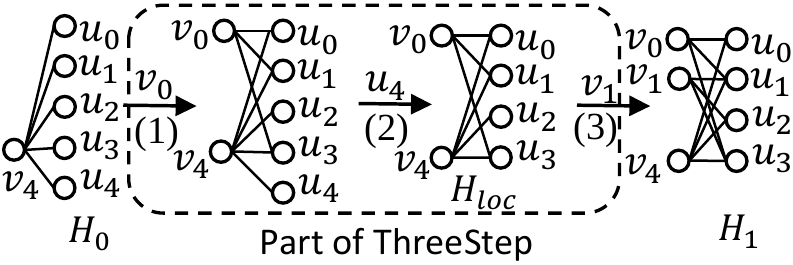}
	\vspace{-0.15in}
	\caption{Illustration of the \texttt{ThreeStep} procedure.}
	\label{fig:three-step}
\end{figure}

{\roundA We adopt a \emph{reverse search} framework which we call \texttt{bTraversal}~\cite{DBLP:journals/jcss/CohenKS08} for enumerating MBPs. \texttt{bTraversal} is a framework for enumerating maximal subgraph structures that satisfy the hereditary property. In the sequel, we review \texttt{bTraversal} in Section~\ref{subsec:bTraversal}.} We observe that \texttt{bTraversal} imposes a requirement that is more demanding than necessary and we relax it to achieve an improved framework called \texttt{iTraversal} in Section~\ref{subsec:iTraversal}. We then develop a series of three techniques, namely left-anchored traversal (Section~\ref{sec:one-side}), right-shrinking traversal (Section~\ref{sec:non-increasing}), and exclusion strategy (Section~\ref{sec:algorithm}) for further boosting \texttt{iTraversal}'s performance. Finally, we present a summary of \texttt{iTraversal} and its running time and delay in Section~\ref{sec:algorithm}.

\subsection{The Basic Framework: \texttt{bTraversal}}
\label{subsec:bTraversal}

The key insight of \texttt{bTraversal} is that given a solution $H$ (which corresponds to a set of vertices), it is possible to find another solution by \emph{excluding} some existing vertices from and \emph{including} some new ones to $H$. 
The rationale is that (1) due to the hereditary property, $H$ is a $k$-biplex so it will still be a $k$-biplex after some vertices are excluded; and (2) after some vertices are excluded from $H$, it becomes possible to include some new vertices to $H$ while retaining the property and thus it finds another solution. 
%

\begin{figure*}
	\hspace*{-.4cm}
	
	\centering
	\begin{tabular}{c c c c c c c}
		\begin{minipage}{3cm}
			\includegraphics[width=3.5cm]{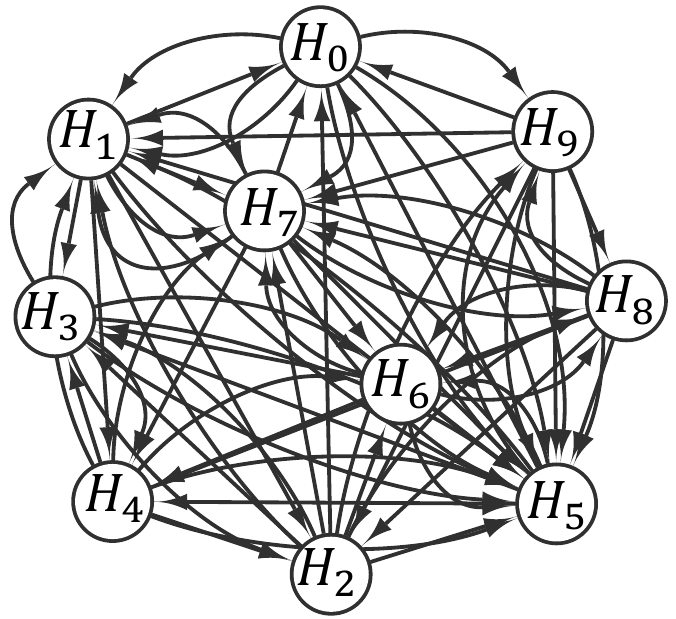}
		\end{minipage}
		& 
		&
		\begin{minipage}{3.5cm}
			\includegraphics[width=3.5cm]{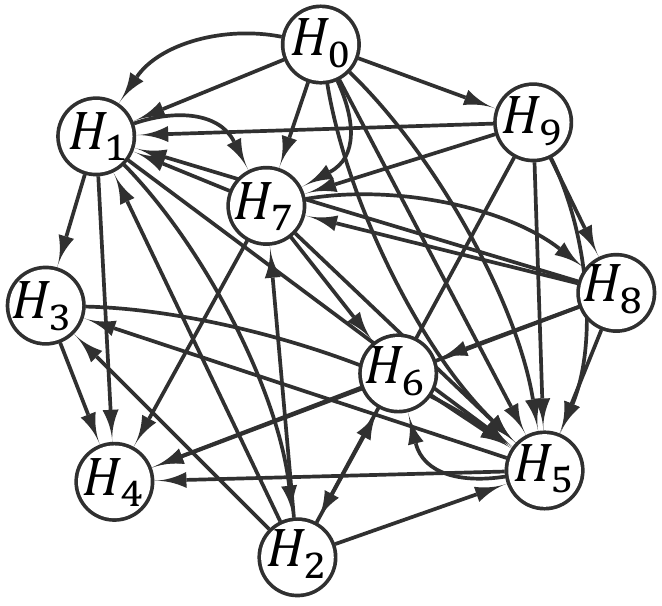}
		\end{minipage}
		&
		&
		\begin{minipage}{3cm}
			\includegraphics[width=3.5cm]{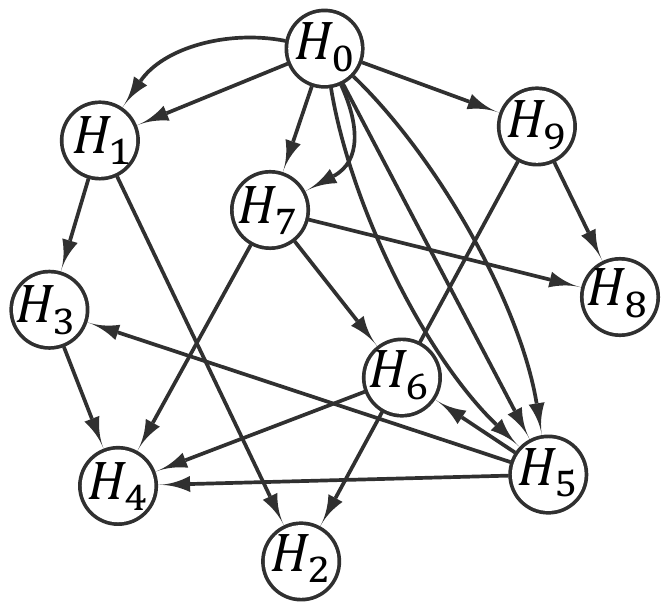}
		\end{minipage}
		& 
		&
		\begin{minipage}{3.5cm}
			\includegraphics[width=3.5cm]{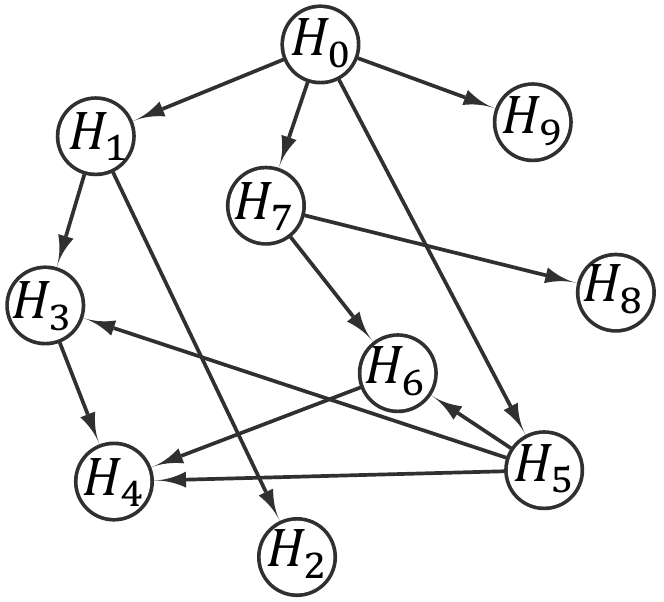}
		\end{minipage}
		\\
		(a) \texttt{bTraversal} ($\mathcal{G}$)
		& 
		&
		(b) Left-anchored traversal ($\mathcal{G}_L$)
		&
		&
		(c) Right-shrinking traversal ($\mathcal{G}_R$)
		& 
		&
		(d) \texttt{iTraversal} ($\mathcal{G}_E$)
	\end{tabular}
	\caption{Solution graphs underlying different algorithms (based on the input graph in Figure~\ref{fig:input_graph}).}
	\label{fig:solution_graph}
\end{figure*}

{Specifically, \texttt{bTraversal} first finds} one solution $H_0$ as the initial one, which can be any one among all solutions. This can be achieved easily, say, by iteratively including vertices to an initially empty set while retaining the $k$-biplex property until it is not possible to do so. {It then} finds solutions from $H_0$ via a procedure of excluding and including vertices from and to $H_0$ and recursively performs the procedure from those solutions found for finding more solutions until no new solutions are found. It uses the following three-step procedure for finding solutions from one solution $H$, which we call \texttt{ThreeStep}.
\begin{itemize}[leftmargin=*]
    {
	\item \textbf{Step 1 (Almost-satisfying graph formation).} 
	For each vertex $v\in V(G) \backslash V(H)$, it forms a new induced subgraph $G[V(H)\cup \{v\}]$ (or simply $G[H\cup v]$) by including $v$ to $H$. We call each such graph $G[H\cup v]$ an \emph{almost-satisfying graph} since it is not a $k$-biplex (since otherwise $H$ is not maximal) and would be so if one vertex, i.e., $v$, is excluded. 

	\item \textbf{Step 2 (Local solution enumeration).} 
	{\roundA For each almost-satisfying graph $G[H\cup v]$, it enumerates all (induced) subgraphs $H_{loc}$ of $G[H\cup v]$ that (1) involve $v$, (2) are $k$-biplexes, and (3) are \emph{maximal} w.r.t. $G[H\cup v]$ (which means that there exists no vertex $u\in V(H)\cup \{v\} \backslash V(H_{loc})$ such that $G[H_{loc}, u]$ is a $k$-biplex).} Essentially, it solves the MBP enumeration problem with the input of $G[H\cup v]$, which should be much easier than the original one with the input $G$. We call such a subgraph $H_{loc}$ a \emph{local solution} since it is maximal locally w.r.t. $G[H\cup v]$ and may not be maximal w.r.t. $G$. 
	We call the procedure of enumerating all local solutions within an almost-satisfying graph $G[H\cup v]$ \texttt{EnumAlmostSat}. In Section~\ref{sec:enumalmostsat}, we present an implementation of \texttt{EnumAlmostSat}.
}
	\item \textbf{Step 3 (Local solution extension).} 
	For each local solution $H_{loc}$, it extends $H_{loc}$ to a real solution $H'$ (i.e., $H'$ is maximal w.r.t. $G$), by iteratively including to $H_{loc}$ vertices from outside $H_{loc}$ until $H_{loc}$ becomes maximal w.r.t. $G$. We remark that during this step, for each local solution $H_{loc}$, it is extended to only one real solution $H'$, e.g., it includes vertices to $H_{loc}$ by following a pre-set order on all vertices. 
\end{itemize}

\begin{example}
Consider the \texttt{ThreeStep} for finding $H_1$ from $H_0$ in Figure \ref{fig:three-step} with $k=1$. We (1) form an almost-satisfying graph $G[H_0,v_0]$ by including $v_0$ to $H_0$, (2) find a local solution $H_{loc}$ by excluding $u_4$ from $G[H_0,v_0]$ and (3) extend $H_{loc}$ to $H_1$ by  including $v_1$ to $H_{loc}$ while retaining the $k$-biplex property.
\end{example}

{\roundA  We summarize the \texttt{bTraversal} algorithm in Algorithm~\ref{alg:btraversal}. A solution may be traversed from multiple solutions. To avoid duplication, a B-tree is used for storing those solutions that have been found, where the key of a solution is specified by the vertices of the solution (Line 1 and 7-8).} 

\begin{algorithm}{}
	\small
	\caption{The algorithm: {\tt \textit{bTraversal}}.}
	\label{alg:btraversal}
	\KwIn{Bipartite graph $G=(\mathcal{L}\cup \mathcal{R},\mathcal{E})$, integer $k\geq 1$;}
	\KwOut{All maximal $k$-biplexes;}
	
	\textbf{Initialize} $H_0$ as any maximal $k$-biplex, B-tree $\mathcal{T}=\{H_0\}$\; 
	\texttt{ThreeStep}$(G, H_0,\mathcal{T})$\;

	\SetKwBlock{Enum}{Procedure \texttt{ThreeStep}$(G, H,\mathcal{T})$}{}
	\Enum{
		(\textbf{Step 1}) \ForEach{ $v$ in $V(G)\backslash V(H)$ }{
		   
		(\textbf{Step 2}) \ForEach{ $H_{loc}$ in \texttt{EnumAlmostSat}$(G[H,v])$}{
		
				(\textbf{Step 3}) 
				Extend $H_{loc}$ to be a maximal $k$-biplex $H'$ with vertices from $V(G)\backslash V(H_{loc})$\;
				\If{$H'\notin \mathcal{T}$}{
					Insert $H'$ to $\mathcal{T}$\; 
					\texttt{ThreeStep}$(G, H',\mathcal{T})$\;
				}	
			}
	}
	}
\end{algorithm}

Suppose we take each solution as a node and create a directed edge from solution $H$ to solution $H'$ if \texttt{bTraversal} can find $H'$ from $H$ via the above three-step procedure. Then, we obtain a graph structure on top of all solutions. This graph structure is called a \emph{solution graph}~\cite{DBLP:conf/stoc/ConteU19}, which we denote by $\mathcal{G}$. {\roundB We note that a solution graph is a multi-graph since from one solution $H$, \texttt{bTraversal} may find another solution $H'$ by forming different almost-satisfying graphs.} We refer to the vertices and directed edges in the solution graph as nodes and links, respectively, and reserve the former notions for those in the graph $G$. Then, \texttt{bTraversal} corresponds to a \emph{depth-first search} (DFS) procedure over the solution graph $\mathcal{G}$. According to~\cite{DBLP:conf/stoc/ConteU19}, the solution graph $\mathcal{G}$ is \emph{strongly connected} and thus \texttt{bTraversal} is able to enumerate all solutions starting from any solution.
{\roundA  To illustrate, consider the input graph in Figure \ref{fig:input_graph} with $k=1$. The corresponding solution graph is shown in Figure \ref{fig:solution_graph}(a), which is strongly connected with 10 solutions and 76 links.}


\subsection{An Improved Framework: \texttt{iTraversal}}
\label{subsec:iTraversal}

\texttt{bTraversal} makes a requirement that any solution is reachable from any other solution in the solution graph (i.e., the solution graph is strongly connected) so that it can find \emph{all} solutions from \emph{any} initial solution. To fulfill this requirement, it would find many solutions from one solution. To see this, consider the above three-step procedure \texttt{ThreeStep}, where $O(|V(G)|)$ almost-satisfying graphs are formed (Step 1), for each almost-satisfying graph, an exponential number of local solutions are enumerated (Step 2), and each local solution is extended to a real solution (Step 3). Consequently, the underlying solution graph would be dense and the DFS procedure on the solution graph would be costly. Note that the time complexity of DFS is proportional to the number of links in the solution graph. 

We observe that the requirement of a strongly connected solution graph by \texttt{bTraversal} is stronger than necessary. In fact, it would be sufficient as long as all solutions are reachable from some \emph{specific} solution since we then can start the DFS procedure from this specific solution and reach all solutions. Motivated by this, in this paper, we propose an improved traversal framework called \texttt{iTraversal}, which performs the DFS procedure from some \emph{specific} but not \emph{arbitrary} solution on a solution graph. With a designated initial solution, \texttt{iTraversal} makes it possible to significantly sparsify the solution graph that is defined by \texttt{bTraversal} while maintaining that all solutions are reachable from the initial solution. 
To illustrate, consider again the example in Figure~\ref{fig:input_graph}. The solution graph that is defined by \texttt{bTraversal} is shown in Figure~\ref{fig:solution_graph}(a), which involves 76 links and is strongly connected. One solution graph that could be defined by \texttt{iTraversal} is shown in Figure~\ref{fig:solution_graph}(d) and the initial solution that could be chosen by \texttt{iTraversal} is $H_0=(v_4,\{u_i\}_{i=0}^4)$. All solutions are reachable from $H_0$ in the sparsified solution graph with 13 links.

One immediate question is: \emph{what is a good initial solution $H_0$ among all possible solutions?} We consider two desiderata: (1) $H_0$ can be computed easily and (2) the solution graph defined by \texttt{bTraversal} can be sparsified significantly (by dropping some links from the solution graph) while all solutions are still kept reachable from $H_0$. 


Our proposal is to use $H_0 = (L_0, \mathcal{R})$ as the initial solution, where $L_0$ is any maximal set of vertices from $\mathcal{L}$ with $(L_0, \mathcal{R})$ being a $k$-biplex~\footnote{An alternative proposal is $H_0 = (\mathcal{L}, R_0)$ that is defined symmetrically and all techniques proposed in this paper would still apply. These two proposals are symmetric and are evaluated empirically in experiments.}. To illustrate, consider the example in Figure~\ref{fig:input_graph} with $k=1$. We obtain $H_0=(L_0,\mathcal{R})$ where $L_0=\{v_4\}$ and $\mathcal{R}=\{u_0,u_1,u_2,u_3,u_4\}$. Next, we explain how $H_0$ meets the two desiderata.

Consider the first desideratum. We can construct $H_0$ easily as follows. First, we initialize $H_0$ as $(\emptyset, \mathcal{R})$. Note that $(\emptyset, \mathcal{R})$ is a $k$-biplex since there are no vertices at the left side (i.e., $\emptyset$) and for each vertex at right side (i.e., $\mathcal{R}$), it disconnects from no vertices from the left side. Second, we extend $H_0$ by iteratively including vertices from $\mathcal{L}$ while maintaining that $H_0$ is a $k$-biplex until this is not possible. At the end, $H_0$ corresponds to a maximal $k$-biplex, i.e., a solution. This process would check for each vertex from $\mathcal{L}$ whether it can be included to $H_0$, which is efficient.

Consider the second desideratum. With $H_0 = (L_0, \mathcal{R})$ as the initial solution, we are able to identify a set of paths to traverse from $H_0$ to all solutions {\roundA since $H_0$ includes $\mathcal{R}$ and can reach every solution by iteratively including vertices from the left side of the target solution and excluding vertices that are not in the target solution.}
Hence, we can drop a large amount of links that do not appear along any of these paths ({\roundA details are in Section~\ref{sec:one-side} and Section~\ref{sec:non-increasing}}). In this way, the solution graph $\mathcal{G}$ could be sparsified significantly while retaining that all solutions are reachable from $H_0$, as shown in Figure \ref{fig:solution_graph}(d).

\subsection{\texttt{iTraversal}: Left-anchored Traversal}
\label{sec:one-side}

%

Let $H'' = (L'', R'')$ be any solution that is different from the initial solution $H_0=(L_0,\mathcal{R})$ and $\mathcal{P} = \langle H_0, H_1, ..., H_n\rangle$ be a path from $H_0$ to $H_n=H''$ in $\mathcal{G}$. Consider the first link among the path $\mathcal{P}$, i.e., $\langle H_0, H_1\rangle$. Recall that when finding $H_1$ from $H_0$ via the procedure \texttt{ThreeStep} in Section~\ref{subsec:bTraversal}, it first \emph{includes} a vertex $v\in V(G) \backslash V(H_0)$ for forming an almost-satisfying graph (Step 1). We observe that this vertex $v$ is always from the left side since $V(G)\backslash V(H_0) = \mathcal{L} \backslash L_0$. We call such a link, which is formed by including a vertex from the left side for forming an almost-satisfying graph in the procedure \texttt{ThreeStep}, as a \emph{left-anchored} link. {\revision We note that left-anchored links are defined based on solutions but not intermediate ones (e.g., local solutions) and each link is either a left-anchored link or a non-left-anchored one.} {\roundA To illustrate, consider the link $\langle H_0,H_1 \rangle$ from $H_0$ to $H_1$ in Figure \ref{fig:three-step}.} It is a left-anchored link 
since the almost-satisfying graph is formed by including vertex $v_0\in \mathcal{L}$.

\begin{figure}[t]
	\centering
	\includegraphics[width=0.82\linewidth]{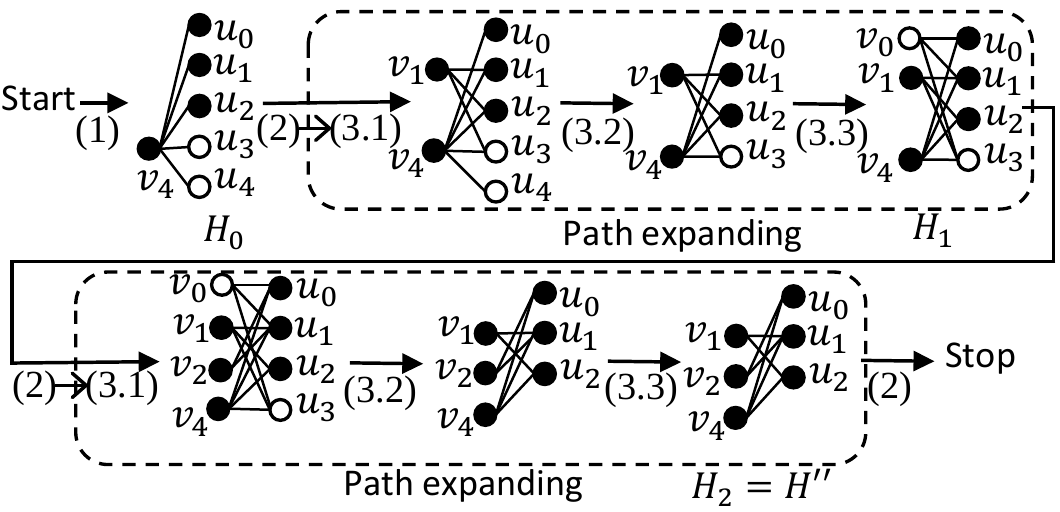}
	\caption{\revision
	Illustration of the four-step procedure of finding a path involving only left-anchored links, i.e., $\mathcal{P}_L(H'')$ (black nodes denote the vertices in the target solution $H''$).}
	\label{fig:four_step_procedure}
\end{figure}

Based on the above discussion, we know that the first link of any path from $H_0$ to $H''$ is a left-anchored link. This triggers the following question: \emph{can we always find a path from $H_0$ to $H''$, which involves left-anchored links only?} The answer is interestingly positive. 
In the following, we present a procedure, which defines for any solution $H''$ a path $\langle H_0,H_1,...,H_n \rangle$ with $H_n=H''$ in $\mathcal{G}$, which traverses from $H_0$ to $H''$ and involves left-anchored links only. We denote this path by $\mathcal{P}_L(H'')$. Specifically, the procedure has four steps and maintains the following invariant (which can be proved by induction).
\begin{equation}
    R'' \subseteq R_i, \text{~~for~~}i = 0, 1, ..., n
\end{equation}
\begin{itemize}[leftmargin=*]
	\item \textbf{Step 1: Path initialization.} Initialize $i$ to be 0. 
	Note that $R'' \subseteq R_i = \mathcal{R}$ (basis step for proving the invariant).
	
	\item \textbf{Step 2: Termination checking.} If $L'' \backslash L_i\! =\! \emptyset$, set $n\!=\!i$ and stop. 
	
	\item \textbf{Step 3: Path expanding.} Find another solution $H_{i+1}$ from $H_i$ via a \emph{left-anchored} link in $\mathcal{G}$ as follows:
		\begin{itemize}
			\item \textbf{Step 3.1.} Pick a vertex $v$ in $L'' \backslash L_i$ and form an almost-satisfying graph $G[H_i, v]$.
			
			\item \textbf{Step 3.2.} Find a local solution $H_{i+1}' = (L_{i+1}', R_{i+1}')$ by extending $((L''\cap L_i) \cup \{v\}, R'')$ to be maximal within $G[H_i, v]$. Note that $((L''\cap L_i) \cup \{v\}, R'')$ is: (1) a subgraph of $G[H_i, v]$ (since $(L''\cap L_i) \subseteq L_i$ and $R''\subseteq R_i$) and (2) a $k$-biplex (since it is subgraph of $H''$).
						
			\item \textbf{Step 3.3.} Extend $H_{i+1}'$ to be a MBP (within $G$), which we denote by $H_{i+1} = (L_{i+1}, R_{i+1})$. Note that $R''\subseteq R_{i+1}' \subseteq R_{i+1}$ (induction step for proving the invariant).
		\end{itemize}
	\item \textbf{Step 4: Repetition.} Increase $i$ by 1 and go to Step 2.
\end{itemize}

\begin{example}
\label{example:four-step}
Given solution $H''\!=\!(L'',R'')$, where $L''\!=\!\{v_1,v_2,v_4\}$ and $R''=\{u_0,u_1,u_2\}$, based on the input graph in Figure \ref{fig:input_graph} with $k=1$. We consider a path from $H_0=(L_0,\mathcal{R})$, where $L_0=\{v_4\}$, to $H''$ formed by the above procedure, as shown in Figure \ref{fig:four_step_procedure}. For the first round, we (1) initialize the path, (2) check $L''\backslash L_0=\{v_1,v_2\}$, (3.1) pick vertex $v_1$ and form an almost-satisfying graph $G[H_0,v_1]$, (3.2) find a local solution $(\{v_1,v_4\},\{u_0,u_1,u_2,u_3\})$ which includes $L''\cap L_0 \cup \{v_1\}=\{v_1,v_4\}$ and $R''$ and (3.3) extend it to a solution $H_1=(L_1,R_1)$ where $L_1=\{v_0,v_1,v_4\}$ and $R_1=\{u_0,u_1,u_2,u_3\}$. We repeat for the second round, (2) check $L''\backslash L_1=\{v_2\}$, (3.1) form an almost-satisfying graph $G[H_1,v_2]$, (3.2) find a local solution $(\{v_1,v_2,v_4\}, \{u_0,u_1,u_2\})$ which includes  $L''\cap L_1 \cup \{v_2\}=\{v_1,v_2,v_4\}$ and $R''$ and (3.3) extend it to a solution $H_2=(L_2,R_2)$ where $L_2=\{v_1,v_2,v_4\}$ and $R_2=\{u_0,u_1,u_2\}$. Finally, we check $L_2\backslash L''=\emptyset$ and get $H_2=H''$.
\end{example}

\begin{lemma}
	\label{lemma:left-anchored-path}
	The procedure of finding the path $\mathcal{P}_L(H'')$ for a given solution $H''$ would always terminate and path $\mathcal{P}_L(H'')$ ends at $H''$, i.e., $H_n = H''$.
\end{lemma} 
\begin{proof}
We first prove that the procedure would always terminate. To this end, we define a similarity measurement between two MBPs. Given two MBPs $H\! =\! (L, R)$ and $H'\! =\! (L', R')$, we define the \emph{similarity} between $H$ and $H'$, denoted by $S(H, H')$, as the number of vertices that are shared by $H$ and $H'$, i.e., 
\begin{equation}
	S(H, H') = |V(H) \cap V(H')| = |L\cap L'| + |R\cap R'|. 
	\label{equation:solution-similarity}
\end{equation}
%
We deduce that $H_{i+1}$ shares at least one more vertex with $H''$ than $H_i$ for $i = 0, 1, ..., n-1$. That is, 
\begin{equation}
S(H_{i+1}, H'')\! \ge\! S(H_i, H'') + 1, \text{~~for~~}i=0, 1, ..., n-1.
\end{equation}
%
%
This is because (1) both $R_{i+1}$ and $R_i$ include $R''$ (based on the invariant of the procedure); and (2) $L_{i+1}$ includes all vertices that are shared by $L_i$ and $L''$ and at least one vertex $v$ from $L''\backslash L_i$.
Therefore, we further deduce that the procedure would always stop since the similarity to $H''$ increases by at least 1 after each round and it is bounded by $|H''|$.

We then prove that $H_n = H''$ by contradiction. Suppose $H_n \neq H''$. {\roundA We deduce that $H''$ would not be an MBP since $H_n = (L_n, R_n)$ is a larger k-biplex containing $H''$, given (1) $H''\subseteq H_n$ (since $L''\backslash L_n = \emptyset$ which means $L'' \subseteq L_n$ and $R''\subseteq R_n$ based on the invariant), (2) $H''\neq H_n$ based on the assumption, and (3) $H_n$ is a $k$-biplex.} This leads to a contradiction.
\end{proof}

In conclusion, we succeed in finding for any solution $H''$ a path that traverses from $H_0$ to $H''$ and involves left-anchored links only in $\mathcal{G}$. Therefore, we propose to \emph{drop} all non-left-anchored links from $\mathcal{G}$. We denote the resulting solution graph by $\mathcal{G}_L$. {\roundA For example, $\mathcal{G}_L$ based on the input graph in Figure~\ref{fig:input_graph} is shown in Figure~\ref{fig:solution_graph}(b), which involves 41 links and all solutions are reachable from $H_0$.} It is clear that a DFS procedure from $H_0$ on $\mathcal{G}_L$, which we call the \emph{left-anchored traversal}, would return all solutions. We present this result in the following lemma.

\begin{lemma}
	\label{lemma:left-anchored-traversal}
	{
	Given a bipartite graph $G=(\mathcal{L}\cup \mathcal{R},\mathcal{E})$ with an initial MBP $H_0 = (L_0,\mathcal{R})$, the left-anchored traversal enumerates all MBPs.
	}
\end{lemma} 
	

\smallskip\noindent\textbf{Remarks.}
{\roundA We remark that the four-step procedure for finding the path $\mathcal{P}_L(H'')$ is a conceptual one for verifying the correctness of the left-anchored traversal only. The implementation of left-anchored traversal will be discussed in Section~\ref{sec:algorithm}.} {\revision In addition, we note that the sparsified solution $\mathcal{G}_L$ is no longer strongly connected as $\mathcal{G}$ does.} To see this, consider that a solution in the form of $(\mathcal{L}, R_0)$ (formed by extending $(\mathcal{L}, \emptyset)$ to be an MBP). There exist no links going from this solution in $\mathcal{G}_L$ since its left side is full. In our experiments, we show that $\mathcal{G}_L$ is significantly sparser than $\mathcal{G}$, e.g., $\mathcal{G}_L$ has about 20$\times$ fewer links than $\mathcal{G}$ on average. 

\subsection{\texttt{iTraversal}: Right-shrinking Traversal}
\label{sec:non-increasing}

\begin{figure}[]
	\centering
	\includegraphics[width=0.82\linewidth]{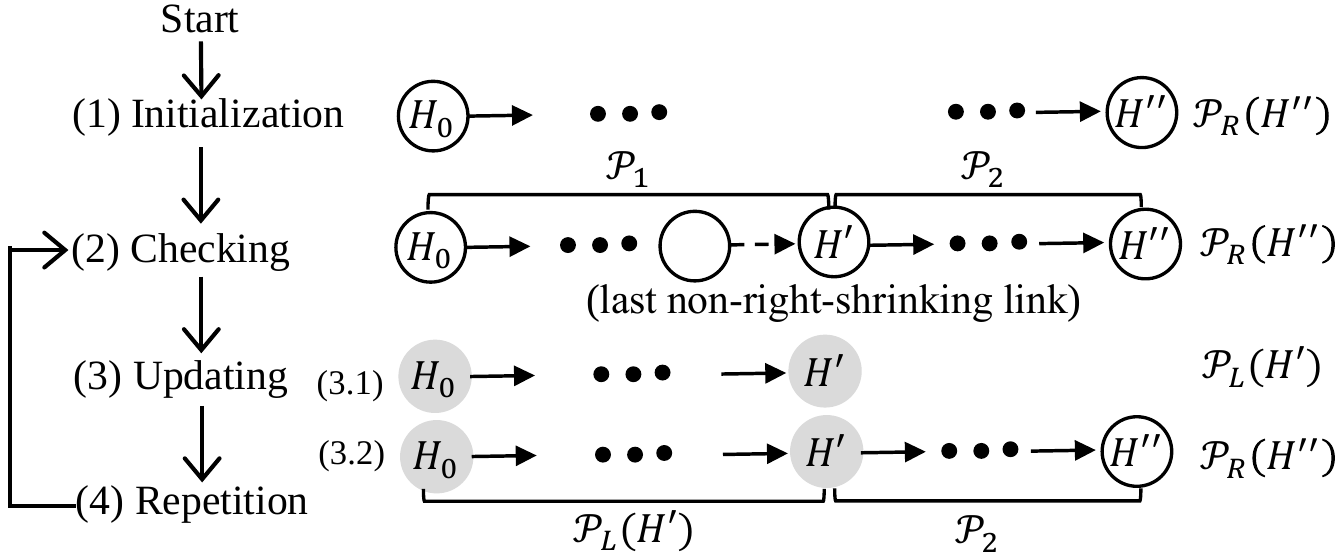}
	\caption{\revision Illustration of the procedure for finding a path involving only right-shrinking links, i.e., $\mathcal{P}_R(H'')$ (the path with gray nodes refers to a new path after updating).}
	\label{fig:shrinking_procedure}
\end{figure}

Consider a link from $H\! =\! (L, R)$ to $H'\! =\! (L', R')$ in $\mathcal{G}_L$. We say it is a \emph{right-shrinking link} iff $R'\! \subseteq\! R$.
Note that we focus on the right side for this definition. 

Consider the path $\mathcal{P}_L(H'')\! =\! \langle H_0, H_1, ..., H_n\rangle$ from $H_0$ to a solution $H''$, where $H'' = H_n$, as defined in Section~\ref{sec:one-side}. We have two observations: (1) the first link (which is from $H_0$ to $H_1$) is a right-shrinking link since $R_1 \subseteq  R_0$ and (2) the last link (which is from $H_{n-1}$ to $H_n$) is also a right-shrinking link since $R_n = R'' \subseteq R_{n-1}$ according to the invariant in Section~\ref{sec:one-side}. This triggers the following question: \emph{can we reach $H''$ from $H_0$ by traversing right-shrinking links only in $\mathcal{G}_L$?} The answer is interestingly also positive. 

In the following, we present a procedure, which defines a path in $\mathcal{G}_L$ for any solution $H''$, which traverses from $H_0$ to $H''$ and involves right-shrinking links only. We denote this path by $\mathcal{P}_R(H'')$. Specifically, the procedure has four steps as follows. A visual illustration of the following procedure is shown in Figure \ref{fig:shrinking_procedure}.
\begin{itemize}[leftmargin=*]
	\item \textbf{Step 1 (Path initialization).} Find the path $\mathcal{P}_L(H'')$ from $H_0$ to $H''$ via the four-step procedure in Section~\ref{sec:one-side} and initialize $\mathcal{P}_R(H'')$ to be $\mathcal{P}_L(H'')$.
	
	\item \textbf{Step 2 (Termination checking).} Check if $\mathcal{P}_R(H'')$ involves only right-shrinking links; If so, stop; otherwise, let $\langle H, H'\rangle$ be the \emph{last} non-right-shrinking link in $\mathcal{P}_R(H'')$, $\mathcal{P}_1$ be the portion from $H_0$ to $H'$ in $\mathcal{P}_R(H'')$, and $\mathcal{P}_2$ be the portion from $H'$ to $H''$ in $\mathcal{P}_R(H'')$. Note that $\mathcal{P}_2$ involves right-shrinking links only.
	
	\item \textbf{Step 3 (Path updating).}	
	\begin{itemize}
		\item \textbf{Step 3.1.} Find the path $\mathcal{P}_{L}(H')$ from $H_0$ to $H'$ via the four-step procedure in Section~\ref{sec:one-side}.

		\item \textbf{Step 3.2.} Update $\mathcal{P}_R(H'')$ by replacing $\mathcal{P}_1$ with $\mathcal{P}_{L}(H')$. Note that both $\mathcal{P}_1$ and $\mathcal{P}_{L}(H')$ start with $H_0$ and end at $H'$.
	\end{itemize}
	
	\item \textbf{Step 4 (Repetition).} Go to Step 2 for another round.
\end{itemize}

\begin{lemma}
	\label{lemma:right-shrinking-path}
	The procedure of finding the path $\mathcal{P}_R(H'')$ for a given solution $H''$ would always terminate and the found $\mathcal{P}_R(H'')$ involves right-shrinking links only.
\end{lemma} 
\begin{proof}
We first prove that the procedure would always terminate with two steps.
First, we show that $H'$ at one round does not appear in $\mathcal{P}_2$ at a previous round by contradiction. Let $\langle H_{cur}, H_{cur}'\rangle$ be the last non-right-shrinking link at the current round and $\langle H_{pre}, H_{pre}'\rangle$ be that at the previous round. Suppose $H_{cur}'$ appears in the path $\mathcal{P}_2$ at the previous round. {\roundA There exist a path from $H_{cur}'$ to $H_{pre}'$
(since $H_{pre}'$ is in the current $\mathcal{P}_2$ from $H_{cur}'$ to $H''$)
and another one from $H_{pre}'$ to $H_{cur}'$ 
(since $H_{cur}'$ is in the previous $\mathcal{P}_2$ from $H_{pre}'$ to $H''$),
both involving right-shrinking links only. We then deduce that $R_{pre}'\subseteq R_{cur}'$ and $R_{cur}'\subseteq R_{pre}'$, which imply that $R_{cur}' = R_{pre}'$. 
In addition, we have $R_{pre}'\subseteq R_{cur}$ based on the invariant in Section~\ref{sec:one-side} over the left-anchored sub-path $\langle H_{cur},...,H_{pre}'\rangle$.}
We deduce that $R_{cur}' = R_{pre}' \subseteq R_{cur}$, which leads to a contradiction to the fact that $\langle H_{cur}, H_{cur}'\rangle$ is non-right-shrinking link. Second, we deduce that $\mathcal{P}_2$ would involve at least one more new solution after each round, which further implies that the above procedure would always terminate since the number of unique solutions is bounded.

We then prove that the path $\mathcal{P}_R(H'')$ involves right-shrinking links only. This can be clearly verified by the Step 2 (Termination checking), i.e., the procedure only stops when $\mathcal{P}_R(H'')$ involves right-shrinking links only.
\end{proof}


In conclusion, we succeed in finding for any solution $H''$ a path that traverses from $H_0$ to $H''$ and involves right-shrinking links only in $\mathcal{G}_L$. Therefore, we propose to \emph{drop} all non-right-shrinking links from $\mathcal{G}_L$. We denote the resulting solution graph by $\mathcal{G}_R$. {\roundA For example, $\mathcal{G}_R$ based on the input graph in Figure~\ref{fig:input_graph} is shown in Figure~\ref{fig:solution_graph}(c), which involves 21 links and all solutions are reachable from $H_0$.}  It is clear that a DFS procedure from $H_0$ on $\mathcal{G}_R$, which we call the \emph{right-shrinking traversal}, would return all solutions. We present this result in the following lemma.

\begin{lemma}
	\label{lemma:right-shrinking-traversal}
	{
	Given a bipartite graph $G=(\mathcal{L}\cup \mathcal{R},\mathcal{E})$ with an initial MBP $H_0 \!=\! (L_0,\mathcal{R})$, the right-shrinking traversal lists all MBPs.
	}
\end{lemma}

\smallskip\noindent\textbf{Remarks.} {\roundA We remark that the four-step procedure for finding the path $\mathcal{P}_R(H'')$ is a conceptual one for verifying the correctness of the right-shrinking traversal only. The implementation of right-shrinking traversal will be discussed in Section~\ref{sec:algorithm}.} In addition, the right-shrinking traversal is on top of the left-anchored traversal. Besides leading to a sparser solution graph, it re-organizes the search space. {\revision To be specific, any solution $H''=(L'',R'')$ reachable from $H=(L,R)$ in $\mathcal{G}_R$ must satisfy $R''\!\subseteq\! R$.} 
One benefit is that it would be natural to impose some size constraints on the MBPs to be enumerated. {\roundA The traversal from solution $H=(L,R)$ can be pruned if $R$ shrinks below the size threshold (details will be presented in Section~\ref{sec:size-constrained})}.

\subsection{\texttt{iTraversal}: Summary and Analysis}
\label{sec:algorithm}

We present the \texttt{iTraversal} algorithm, which employs left-anchored traversal and right-shrinking traversal in Algorithm~\ref{alg:itraversal}. \texttt{iTraversal} differs from \texttt{bTraversal} in the following aspects. \underline{First}, it takes $(L_0, \mathcal{R})$ but not an arbitrary MBP as the initial solution, where $L_0$ is a maximal subset of $\mathcal{L}$ such that $(L_0, \mathcal{R})$ is a MBP (Line 1). \underline{Second}, in Step 1 of forming almost-satisfying graphs, it \emph{prunes} those vertices in $\mathcal{R}$ from consideration so that it would traverse along left-anchored links only (Line 5). This implements the left-anchored traversal. \underline{Third}, in Step 2 of enumerating local solutions, it \emph{prunes} those local solutions $H_{loc}$ for which there exists a vertex $u\in \mathcal{R}$ such that $u$ is not in $H_{loc}$ and $H_{loc}\cup \{u\}$ is a $k$-biplex (Line 7). These local solutions can be pruned since they can be extended to solutions with the right side containing a vertex that is not contained by the right side of the current solution, i.e., the links from the current solution to these solutions are non-right-shrinking links. This implements the right-shrinking traversal.
%

\begin{algorithm}{}
	\small
	\caption{The algorithm: {\tt \textit{iTraversal}}.}
	\label{alg:itraversal}
	\KwIn{Bipartite graph $G=(\mathcal{L}\cup \mathcal{R},\mathcal{E})$, integer $k\geq 1$ ;}
	\KwOut{All maximal $k$-biplexes;}
	\textbf{Initialize} $H_0=(L_0,\mathcal{R})$, B-tree $\mathcal{T}=\{H_0\}$\;
	\texttt{iThreeStep}$(G, H_0,\mathcal{T})$\;

	\SetKwBlock{Enum}{Procedure \texttt{iThreeStep}$(G, H,\mathcal{T})$}{}
	\Enum{
		(\textbf{Step 1}) \ForEach{ $v$ in $V(G)\backslash V(H)$}{
		    \lIf{$v$ in $\mathcal{R}$}{
		        \textbf{Continue;} //Left-anchored traversal
		    }
		(\textbf{Step 2}) \ForEach{ $H_{loc}$ in \texttt{EnumAlmostSat}$(G[H,v])$}{
		        \lIf{there exists $u$ in $\mathcal{R}\backslash V(H_{loc})$ s.t. $G[H_{loc},\{u\}]$ is $k$-biplex}{
		        \textbf{Continue;} //Right-shrinking traversal}
				(\textbf{Step 3}) 
				Extend $H_{loc}$ to be a maximal $k$-biplex $H'$ with vertices from $V(G)\backslash V(H_{loc}) \backslash \mathcal{R}$\;
				\If{$H'\notin \mathcal{T}$}{
				    Insert $H'$ to $\mathcal{T}$\; 
				    \texttt{iThreeStep}$(G, H',\mathcal{T})$\;
				}	
			}
	}
	}
	
\end{algorithm}
\smallskip
\noindent\textbf{Remark.} For \texttt{bTraversal}, it can be further enhanced with a so-called \emph{exclusion} strategy~\cite{DBLP:conf/sigmod/BerlowitzCK15}. {\revision The idea is to maintain for each solution an \emph{exclusion} set once the solution is traversed and then prune the links towards those solutions which involve a vertex in the exclusion set. Details are referred to the technical report~\cite{TR}. We verify that this strategy is applicable to \texttt{iTraversal}, for which the correctness proof is included in the  technical report~\cite{TR} for the sake of space.} In conclusion, \texttt{iTraversal} implements left-anchored traversal, right-shrinking traversal and the exclusion strategy. The solution graph underlying \texttt{iTraversal} is denoted by $\mathcal{G}_E$, which is even sparser than $\mathcal{G}_R$. 
For example, $\mathcal{G}_E$ based on the input graph in Figure~\ref{fig:input_graph} is shown in Figure~\ref{fig:solution_graph}(d), which involves 13 links and all solutions are reachable from $H_0$.

\smallskip
\noindent\textbf{Total running time}. 
Let $\alpha$ be the number of solutions.
The time cost of \texttt{iTraversal} is dominated by that of calling the \texttt{iThreeStep} procedure $\alpha$ times, each when a solution is found for the first time.
Consider the time cost of the \texttt{iThreeStep} procedure.
Let $\beta$ be the time complexity of the \texttt{EnumAlmostSat} procedure and $\gamma$ be the number of local solutions returned by the \texttt{EnumAlmostSat} procedure.
The \texttt{EnumAlmostSat} procedure is called $O(|\mathcal{L}|)$ times, costing $O(|\mathcal{L}|\cdot \beta)$ time. 
{\roundA 
There are $O(|\mathcal{L}|\cdot \gamma)$ local solutions and for each, the cost is the sum of $O(|\mathcal{R}|\cdot|H_{max}|)$ (for Line 7),  $(|\mathcal{L}|\cdot |H_{max}|)$ (for Line 8), and $O(\log \alpha \cdot |H_{max}|)$ (for Line 9-11), where $H_{max}=(L_{max},R_{max})$ is the solution with the maximum size. 
%
%

Therefore, the overall time complexity of \texttt{iTraversal} is $O(\alpha\cdot (|\mathcal{L}| \cdot \beta + |\mathcal{L}| \cdot\gamma \cdot (|V(G)|\cdot |H_{max}| + \log \alpha \cdot |H_{max}|)))$. Here, $\alpha$ is exponential w.r.t. the bipartite graph size. {\roundF $\beta=O((|L_{max}|\cdot|R_{max}|)^{k+1})$ and $\gamma=O((|L_{max}|\cdot|R_{max}|)^{k})$ are polynomial with $k$ as a constant and will be discussed in Section~\ref{sec:enumalmostsat}. 
In a simpler form, the time complexity is $O(\alpha \cdot|\mathcal{L}|\cdot(|L_{max}|\cdot|R_{max}|)^{k+1} \cdot|V_G|\cdot|H_{max}|)$.}
}

\smallskip
\noindent\textbf{Delay}. The delay of an enumeration algorithm corresponds to the maximum of three parts, namely (1) the time spent after the algorithm starts and before the first solution is found, (2) the time spent between any two consecutive solutions are found, and (3) the time spent after the last solution is found and till the algorithm terminates. 
%
With a small trick~\cite{takeaki2003two}, i.e., we print a solution before and after the recursive call (Line 11 of Algorithm~\ref{alg:itraversal}) in an alternating manner during the sequence of recursive calls, the algorithm would output at least one solution every two successive recursive calls of the \texttt{iThreeStep} procedure. {\roundB Therefore, the delay corresponds to the time complexity of the \texttt{iThreeStep} procedure, i.e., $O\big(|\mathcal{L}| \cdot \beta + |\mathcal{L}| \cdot\gamma \cdot (|V(G)|\cdot |H_{max}| + \log \alpha \cdot |H_{max}|))\big)$, which is polynomial with $k$ as a constant ($\beta$ and $\gamma$ are polynomial, which will be discussed in Section~\ref{sec:enumalmostsat}).
%
We remark that (1) \texttt{iTraversal} improves the delay of \texttt{bTraversal} based on the same implementation of \texttt{EnumAlmostSat}, which is $O\big(|V(G)| \cdot \beta + |V(G)| \cdot\gamma \cdot (|V(G)|\cdot |H_{max}| + \log \alpha \cdot |H_{max}|))\big)$; and (2) {\roundG\texttt{iMB}} and the graph inflation based algorithm \texttt{FaPlexen} have their delay exponential w.r.t. the size of the bipartite graph~\cite{yu2021efficient,DBLP:conf/aaai/ZhouXGXJ20}. }

\section{The \texttt{EnumAlmostSat} Procedure}
\label{sec:enumalmostsat}
Let $(L, R)$ be a solution and $(L\cup\{v\}, R)$ be an almost-satisfying graph, where $v\in \mathcal{L} \backslash L$. The \texttt{EnumAlmostSat} procedure is to enumerate all local solutions within $(L\cup \{v\}, R)$, which involve $v$. 
Essentially, the task is to explore a search space of $\{(L', R')\}$, where $L'\subseteq L$ and $R'\subseteq R$, and find those $(L', R')$'s such that $(L'\cup\{v\}, R')$ is a local solution. Note that this search space has a size of $O(2^{|L|+|R|})$, and hence simply enumerating each pair $(L', R')$ would be costly.


In this paper, we develop a series of techniques for refining the enumerations on $L$ and $R$ so as to reduce and/or prune the search space (Section~\ref{subsec:refine-y-1}, \ref{subsec:refine-y-2}, \ref{subsec:refine-x-2}, and \ref{subsec:refine-x-1}).
We finally present an algorithm based on these refined enumerations and analyze its time complexity (Section~\ref{subsec:enumalmostsat}).

\subsection{Refined Enumeration on $R$: 1.0}
\label{subsec:refine-y-1}

{\roundC We start with an observation presented in the following lemma {\revision (the proofs of the lemmas in this section are presented in the technical report~\cite{TR})}.}

\begin{lemma}
	Given an almost-satisfying graph $(L\cup\{v\}, R)$, each vertex $u\in R$ that connects $v$ is involved in all local solutions within $(L\cup\{v\}, R)$.
	\label{lemma:refine-y-1}
\end{lemma}

Based on Lemma~\ref{lemma:refine-y-1}, we can partition $R$ into two sets, namely one containing those vertices that connect $v$ and the other containing the remaining vertices. We denote the former by $R_{keep}$ and the latter by $R_{enum}$.
{\ChengCommentB An illustration is shown in Figure~\ref{fig:set_partition}(a).}
Then, $R_{keep}$ is involved in all local solutions within $(L\cup \{v\}, R)$ and hence the enumeration of the vertices in $R_{keep}$ can be avoided when enumerating $R'\subseteq R$. Specifically, when enumerating the subsets of $R$, we enumerate $R''\subseteq R_{enum}$ only and for each $R''$, we construct a $R'$ as $R'' \cup R_{keep}$. In addition, we only need to enumerate those $R''$'s with $|R''| \le k$ since otherwise $v$ would disconnect more than $k$ vertices in $R'$ and $(L'\cup\{v\}, R')$ cannot be a $k$-biplex. With this, the number of subsets of $R$ to enumerate is reduced from $O(2^{|R|})$ to $O({|R_{enum}|}^k)$.

\subsection{Refined Enumeration on $R$: 2.0}
\label{subsec:refine-y-2}

{\roundC
Based on Section~\ref{subsec:refine-y-1}, the search space is reduced from one containing all subsets of $R$ to one containing all subsets $R''$ of $R_{enum}$ with $|R''|\leq k$. In this section, we further prune some subsets $R''$ of $R_{enum}$ by refining the enumeration on $R_{enum}$. 

Specifically, we partition $R_{enum}$ into two groups, namely $R_{enum}^1 \!=\! \{u\in R_{enum} \mid  \overline{\delta}(u, L) \le k-1\}$ and $R_{enum}^2 = R_{enum} \backslash R_{enum}^1= \{u\in R_{enum} \mid  \overline{\delta}(u, L) =k\}$. Instead of enumerating $R'' \subseteq R_{enum}$ with $|R''|\le k$ directly, we enumerate $R_1''\subseteq R_{enum}^1$ and $R_2''\subseteq R_{enum}^2$ with $|R_1'' \cup R_2''| \le k$ and construct $R''$ as $R_1'' \cup R_2''$. 
{\ChengCommentB An illustration is shown in Figure~\ref{fig:set_partition}(a).}
We then have the following lemma for pruning some enumerations of $R_1''$ and $R_2''$.



\begin{lemma}
    \label{lemma:refine-y-2}
   Let $R' = R'' \cup R_{keep}$ and $R'' = R_1'' \cup R_2''$, where $R_1''\subseteq R_{enum}^1$, $R_2''\subseteq R_{enum}^2$, and $|R_1'' \cup R_2''| \le k$. There does not exist a subset $L'$ of $L$ such that $(L'\cup \{v\}, R')$ is a local solution if (1) $|R_1''\cup R_2''| < k$ and (2) $R_{enum}^1 \backslash R_1'' \neq \emptyset$.
\end{lemma}
Based on Lemma~\ref{lemma:refine-y-2}, we can prune those enumerations of $R_1''$ and $R_2''$ with $|R_1''\cup R_2''| < k$ and $R_{enum}^1\backslash R_1'' \neq \emptyset$, and thus it reduces the enumerations of subsets $R''$ of $R_{enum}$.
}

\begin{figure}[]
	\centering
	\centering
	\begin{tabular}{l c r}
		\begin{minipage}{3.90cm}
			\includegraphics[width=4.1cm]{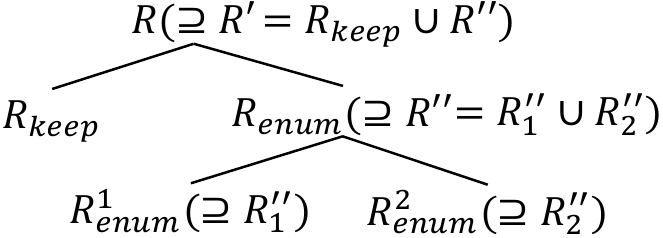}
		\end{minipage}
		&  &
		\begin{minipage}{3.00cm}
			\includegraphics[width=2.8cm]{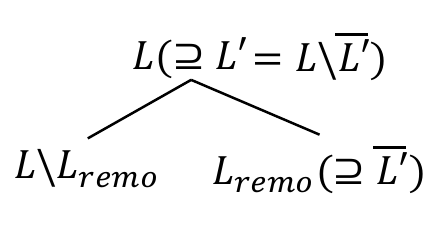}
		\end{minipage}	
		\\
		(a) Partitions on $L$
		&  &
		(b) Partitions on $R$
	\end{tabular}
	\caption{\roundC Illustration of partitions on $L$ and $R$.}
	\label{fig:set_partition}
\end{figure}

\subsection{Refined Enumeration on $L$: 1.0}
\label{subsec:refine-x-1}

Given a subset $R'$ of $R$, where $R' = R'' \cup R_{keep}$ and $R'' = R''_1 \cup R''_2$, we enumerate the subsets of $L$ and for each subset $L'$, construct $(L'\cup \{v\}, R')$ as a candidate of a local solution. A straightforward method is to enumerate all subsets of $L$. However, the search space would be $O(2^{|L|})$. In this section, we aim to refine the enumeration on $L$ so that the search space is reduced.

{\roundC 
We have an observation presented in the following lemma.}
{\roundB
\begin{lemma}
    \label{lemma:refine-x-1}
     In the subgraph $(L'\cup \{v\}, R')$, where $L' = L$ and $R' = R_{keep} \cup R''_1 \cup R''_2$, all vertices except for those in $R''_2$ disconnect at most $k$ vertices and vertices in $R''_2$ disconnect exactly $(k+1)$ vertices.
\end{lemma}

Based on Lemma~\ref{lemma:refine-x-1}, we can construct potential local solutions $(L'\cup \{v\}, R')$ by removing a \emph{minimal} set $\overline{L'}$ of vertices from $L$ so that all vertices in $R''_2$ would disconnect at most $k$ remaining vertices, where $L' = L \backslash \overline{L'}$. Note that (1) after removing $\overline{L'}$ from $L$, those remaining vertices in $L\cup \{v\}$ (i.e., the vertices in $L'\cup \{v\}$) and vertices in $R_{keep} \cup R''_1$ would still disconnect at most $k$ vertices due to the hereditary property and 
(2) $\overline{L'}$ needs to be \emph{minimal} so that $(L'\cup \{v\}, R')$ would be a \emph{maximal} $k$-biplex. 

In addition, when a vertex $v'$ that connects all vertices in $R''_2$ is removed, each vertex in $R''_2$ would still disconnect $(k+1)$ vertices. Therefore, we focus on the set $L_{remo}$ of the vertices that disconnect at least one vertex in $R''_2$, i.e., $L_{remo} = \{v'\in L \mid   \overline{\delta}(v', R''_2) > 0\}$, and enumerate subsets $\overline{L'}$ of $L_{remo}$ to be removed. 
%
%
%
}
{\ChengCommentB An illustration is shown in Figure~\ref{fig:set_partition}(b).}
{\roundB 
For each minimal set $\overline{L'}$ such that $(L'\cup\{v\}, R')$ is a local solution, where $L' = L\backslash \overline{L'}$, it involves no more than $|R''_2|$ vertices. This is because only vertices in $R''_2$ disconnect more than $k$ vertices (exactly $(k+1)$ vertices) and to make one vertex in $R''_2$ disconnect at most $k$ vertices, it is enough to remove one vertex from $L_{remo}$}.


In conclusion, we enumerate $\overline{L'}$ from $L_{remo}$ with $|\overline{L'}| \le |R''_2| \le k$ to be removed, and for each $\overline{L'}$, we construct $(L'\cup \{v\}, R')$ as a candidate of local solution, where $L' = L\backslash \overline{L'}$. In this way, the search space of enumerating $L$ is reduced from $O(2^{|L|})$ to $O(|L_{remo}|^k)$.


\subsection{Refined Enumeration on $L$: 2.0}
\label{subsec:refine-x-2}

When enumerating $\overline{L'}$ from $L$ to be removed, we follow an increasing order of $|L|$ from $0$ to $|R''_2|$. If for a subset $\overline{L'}$ and $L' = L\backslash \overline{L'}$, $(L' \cup \{v\}, R')$ is a local solution, we prune all supersets of $\overline{L'}$ from being enumerated to be removed since removing each of them would construct a $k$-biplex that is not maximal within $(L\cup\{v\},R)$. 


\subsection{\texttt{EnumAlmostSat}: Summary and Analysis}
\label{subsec:enumalmostsat}
We present the algorithm for the \texttt{EnumAlmostSat} procedure in Algorithm~\ref{alg:enumalmostsat}, which is based on the 2.0 versions for enumerating $R$ and $L$. 
{ We use $Powerset(L_{remo})$ to denote the set of all subsets of $L_{remo}$ (Line 4 and 7).} {\revision Note that output solutions in {\ChengCommentC line 6} can be non-globally optimal since an almost-satisfying graph {rather than} the whole graph is used as the context for checking the maximality.}
The correctness can be guaranteed by the lemmas and discussions in this section.

\begin{algorithm}
	\small
	\caption{The algorithm: {\tt \texttt{EnumAlmostSat}}.}
	\label{alg:enumalmostsat}
	\KwIn{Almost-satisfying graph $(L\cup\{v\},R)$;}
	\KwOut{All local solutions including $v$ within $(L\cup\{v\},R)$;}
	
	Partition $R$ into $R_{keep}$, $R_{enum}^1$ and $R_{enum}^2$ (Section~\ref{subsec:refine-y-1} and \ref{subsec:refine-y-2})\;
	\ForEach{$R' = R_{keep} \cup R''_1\cup R''_2$ (as constructed in Section~\ref{subsec:refine-y-2})}{
	    Construct $L_{remo}$ as $\{v'\in L \mid  \overline{\delta}(v', R''_2) > 0\}$\;
	    \ForEach{$\overline{L'} \in Powerset({L_{remo}})$ with $|\overline{L'}|\! \le\! |R''_2|$ in an ascending order of $|\overline{L'}|$}{
	        Construct $L'$ as $L\backslash \overline{L'}$\;
	        \If{$(L'\cup\{v\}, R')$ is a local solution}{
	            Prune those subsets from $Powerset({L_{remo}})$, which are supersets of $\overline{L'}$\;
	                Print $(L'\cup\{v\}, R')$\;
	        }
	    }
	}
\end{algorithm}


\noindent\textbf{Time complexity.} 
The time complexity of \texttt{EnumAlmostSat} is dominated by that of enumerating subsets of $R$ and $L$ (Line 2-8).
There are $O(|R_{enum}|^k)$ enumerations of $R''$ (Line 2), and for each one, it incurs the following costs. First, it partitions $L$, which takes $O(|L|\cdot |R|)$ time (Line 3). Second, it enumerates subsets $\overline{L'}$ of $L_{remo}$ with the size at most $|R_2''|$ on $L_{remo}$ (Line 4). There are $O(|L_{remo}|^k)$ such $\overline{L'}$'s, and for each $\overline{L'}$, the cost is dominated by that of checking the maximality (Line 6), which is $O(|L|\cdot |R|)$. 
{
In summary, \texttt{EnumAlmostSat} would return at most $\gamma=O(|R_{enum}|^k\cdot |L_{remo}|^k)$ local solutions, and the time complexity of \texttt{EnumAlmostSat} is $\beta=O(|R_{enum}|^k\cdot |L_{remo}|^k \cdot |L|\cdot|R|)$, which is polynomial with $k$ as a constant.
}

\section{Extensions of \texttt{iTraversal}}
\label{sec:size-constrained}

In some scenarios, one may want to impose some size constraints on one side or both sides of a MBP to be enumerated. For example, one is interested in only those MBPs with the size on either side to be at least a threshold $\theta$. We show that such constraints can be conveniently incorporated to the \texttt{iTraversal} algorithm so that not all MBPs need to be enumerated, which achieves better efficiency. In contrast, for \texttt{bTraversal}, these constraints cannot be incorporated  easily and all MBPs need to be enumerated and then a filtering step as post-processing is necessary, which is inefficient. For illustration, we consider the constraint that a MBP has the sizes on its both sides at least $\theta$, and we call such a MBP a \emph{large MBP}. However, the techniques can be easily customized for slightly different constraints such that a MBP has the size on its left or right side at least a threshold $\theta$.

\smallskip\noindent\textbf{Right-side pruning.}
Recall that \texttt{iTraversal} adopts the right-shrinking traversal, i.e., it traverses along links from a solution $H = (L, R)$ to another $H' = (L', R')$ with $R' \subseteq R$. This would allow several opportunities for pruning. \underline{(1) Almost-satisfying graph pruning}: When an almost-satisfying graph $G[H, v]$ is formed in Step 1 of \texttt{iThreeStep}, we prune it if $\delta(v,R) + k < \theta$. This is because any solution found based on this almost-satisfying graph would involve $v$ at the left side and less than $\theta$ vertices at the right side (including at most $\delta(v,R)$ vertices that connect $v$ and $k$ vertices that disconnect $v$). \underline{(2) Local solution pruning}: When enumerating the local solutions of an almost-satisfying graph $G[(L,R), v]$ in Step 2 of \texttt{iThreeStep}, the enumerations of the subsets $R'$ of $R$ with $|R'| < \theta$ can be pruned (correspondingly, the enumerations of the subsets $L'$ of $L$ can be saved) since all solutions formed from these local solutions would have a size on their right side being smaller than $\theta$ (due to the right-shrinking traversal). \underline{(3) Solution pruning}: When a solution $H = (L, R)$ with $|R| < \theta$ is found, we skip the recursive call of \texttt{iThreeStep} from $H$ since all solutions that will be found via this call would have the size of their right side at most $|R|$ and smaller than $\theta$ (due to the right-shrinking traversal). 


\smallskip\noindent\textbf{Left-side pruning.} 
Based on the exclusion strategy, we derive the following pruning technique. When a solution $H$ is found, we check whether $|\mathcal{L}|-|\mathscr{E}(H)|<\theta$, where $\mathscr{E}(H)$ is the exclusion set maintained for $H$ by the exclusion strategy. If so, we skip the recursive call \texttt{iThreeStep} from $H$ so that all links from $H$ are pruned. To verify, consider a solution $H'=(L',R')$ that can be reached from $H$ via the call \texttt{iThreeStep} from $H$. {\roundB There are two cases. \underline{Case 1:} $L'\cap \mathscr{E}(H)\neq \emptyset$. The link from $H$ to $H'$ would be pruned by the exclusion strategy. \underline{Case 2:} $L'\cap \mathscr{E}(H) = \emptyset$. We have $|L'| \le   |\mathcal{L}|-| \mathscr{E}(H)| < \theta$, and thus $H'$ cannot be a large MBP.}

\section{Experiments}
\label{sec:exp}

\textbf{Datasets.} We use both real and synthetic datasets. The real datasets are taken from various domains {\roundD(http://konect.cc/)} and summarized in Table \ref{tab:dataset}. 
The synthetic datasets are Erd{\"o}s-R{\'e}yni (ER) graphs, which are generated by first creating a certain number of vertices and then randomly creating a certain number of edges between pairs of vertices. {\roundB We set the number of vertices and edge density as 100k and 10 for synthetic datasets, respectively, by default. Here, the edge density of a graph $G=(\mathcal{L}\cup \mathcal{R},\mathcal{E})$ is defined as $|\mathcal{E}|/(|\mathcal{L}|+|\mathcal{R}|)$.}

\begin{table}[t]
	\centering
	\scriptsize
	\small
	\caption {Real datasets}
	\label{tab:dataset}
	\begin{tabular}{c|c|r|r|r}
		\hline
		 \multicolumn{1}{c|}{\textbf{Name}}
		 & \multicolumn{1}{c|}{\textbf{Category}}
		 & \multicolumn{1}{c|}{$|\mathcal{L}|$}
		 & \multicolumn{1}{c|}{$|\mathcal{R}|$}
		
		& \multicolumn{1}{c}{$|\mathcal{E}|$}
		\\
		\hline\hline
		Divorce   &   HumanSocial   &9  &50     &  225  \\
		\hline
		Cfat     & Miscellaneous &100 &100     &  802 \\
		\hline
		Crime     & Social & 551  & 829   &  1,476 \\
		\hline
		Opsahl     & Authorship &2,865  &4,558       &  16,910 \\
		\hline\hline
		Marvel     & Collaboration  &19,428  &6,486      &  96,662 \\
		\hline
		Writer & Affiliation &89,356  &46,213  & 144,340 \\
		\hline
		Actors & Affiliation &392,400  &127,823    & 1,470,404 \\
		\hline
		IMDB &Communication &428,440  &896,308    & 3,782,463 \\
		\hline
		DBLP & Authorship & 1,425,813  &4,000,150   & 8,649,016 \\
		\hline
		Google & Hyperlink &17,091,929  &3,108,141   & 14,693,125 \\
		\hline
	\end{tabular}
\end{table}

\smallskip
\noindent\textbf{Algorithms.} We compare our algorithm \texttt{iTraversal}, with three baselines, namely, \texttt{iMB} \cite{yu2021efficient}, \texttt{FaPlexen} \cite{DBLP:conf/aaai/ZhouXGXJ20} and \texttt{bTraversal}.
\texttt{iMB} was proposed for enumerating maximal $k$-biplexes. \texttt{FaPlexen} is the state-of-the-art algorithm for enumerating large maximal $k$-plexes and we run it on the inflated graph of a bipartite graph for enumerating maximal $k$-biplexes. \texttt{bTraversal} follows the conventional framework of reverse search and implements the \texttt{EnumAlmostSat} procedure by first inflating the graph and then using an existing algorithm for enumerating local maximal $k$-plexes~\cite{DBLP:conf/sigmod/BerlowitzCK15}.

\smallskip
\noindent\textbf{Settings.} 
All algorithms were implemented in C++ and executed on a machine with a 2.66GHz processor and 32GB of memory, with CentOS installed. 
We set the running time limit (INF) as 24 hours and the memory budget (OUT) as 32GB. The source codes are available at https://github.com/KaiqiangYu/SGMOD22-MBPE.

\subsection{Comparison among Algorithms}
\label{subsec:comparison-results}

\noindent{\textbf{Results of running time (real datasets).}}
{\roundD The results of running time when returning the first 1,000 MBPs (by following the existing studies~\cite{DBLP:conf/sigmod/BerlowitzCK15}) are shown in Figure~\ref{fig:real}(a) for varying datasets, (b) and (c) for varying $k$, and (d) and (e) for varying the number of returned MBPs (results on Writer and DBLP are shown only and those on other datasets show similar clues and {\roundB are put in the technical report \cite{TR})}.} We have the following observations. First, \texttt{iTravesal} outperforms all other algorithms and can finish on all datasets. Second, \texttt{iMB} cannot finish {\roundE within INF} on large datasets due to its ineffective pruning techniques. Third, \texttt{FaPlexen} cannot finish {\roundE within OUT} on large datasets either, mainly due to its underlying graph inflation procedure. For example, for Marvel which has 96 thousand edges, it would generate more than 200 million edges after graph inflation. Therefore, we exclude \texttt{iMB} and \texttt{FaPlexen} for the remaining experiments of enumerating MBPs. Fourth, \texttt{bTraversal} is more scalable than \texttt{iMB} and \texttt{FaPlexen} but still cannot finish on the Google dataset {\roundE within INF}. Fifth, \texttt{iTraversal} and \texttt{bTraversal} have their running time increase as $k$ and the number of MBPs grows, and \texttt{iTraversal} is up to four orders of magnitude faster than \texttt{bTraversal}. {\revision We note that \texttt{iTraversal} has its running time clearly rising as $k$ increases when compared with \texttt{bTraversal}. The reason is that \texttt{iTraversal} adopts the proposed \texttt{L2.0+R2.0}, whose pruning power gets weaker as $k$ increases, while \texttt{bTraversal} adopts a procedure based on graph inflation, which is not that sensitive to $k$.}
\begin{figure}[]
	\centering
	\centering
	\begin{minipage}{8cm}
			\includegraphics[width=8cm]{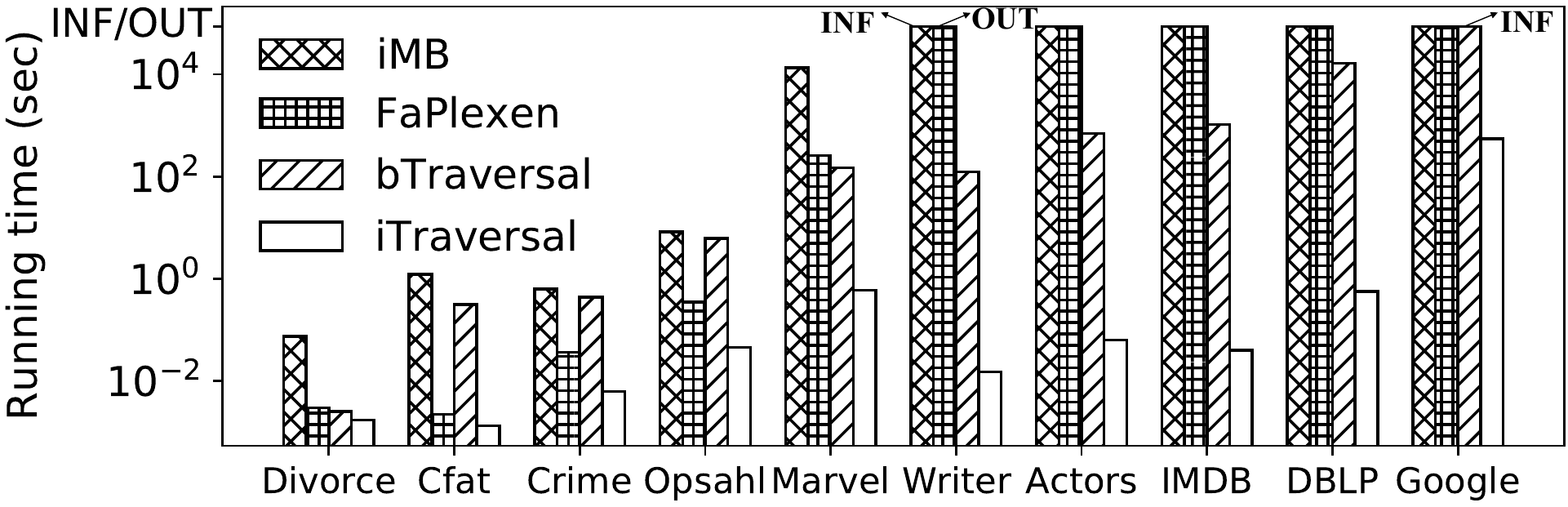}
	\end{minipage}
	\\
	\ \ \ \ \ (a) Varying datasets ($k=1$)
	\begin{tabular}{c c}
		\begin{minipage}{3.50cm}
			\includegraphics[width=4.1cm]{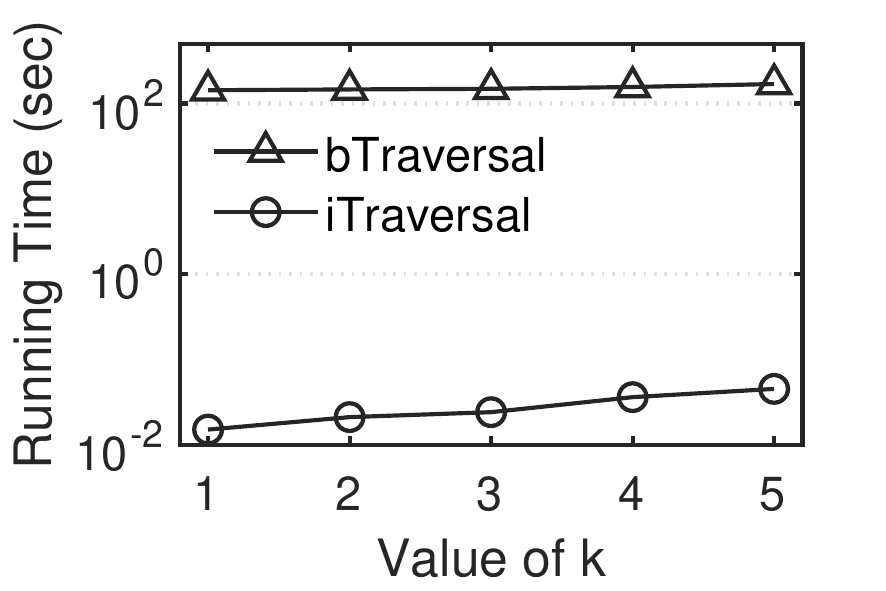}
		\end{minipage}
			&
	\begin{minipage}{3.50cm}
		\includegraphics[width=4.1cm]{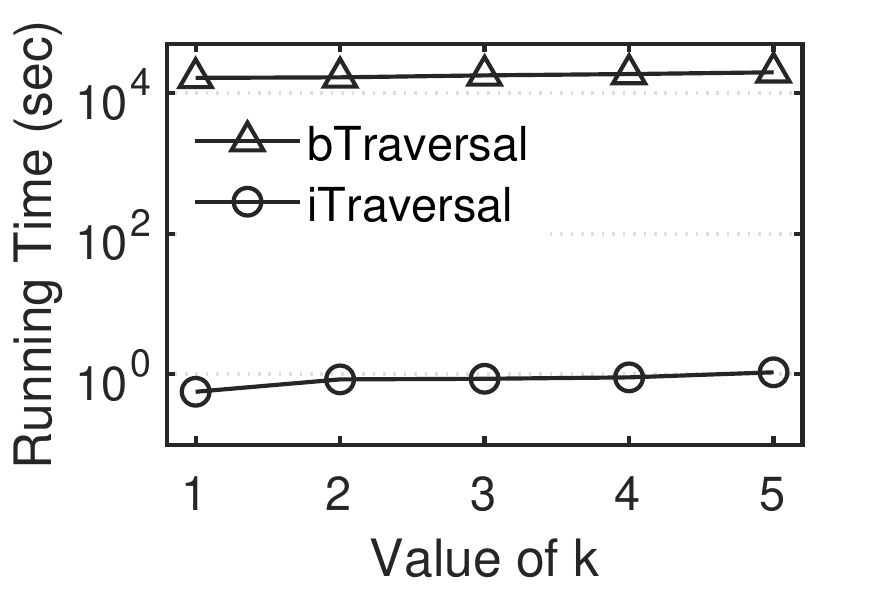}
	\end{minipage}
		\\
		\ \ \ \ (b) Varying $k$ (Writer)
		&
		\ \ \ \ (c) Varying $k$ (DBLP)
		\\
		\begin{minipage}{3.50cm}
			\includegraphics[width=4.1cm]{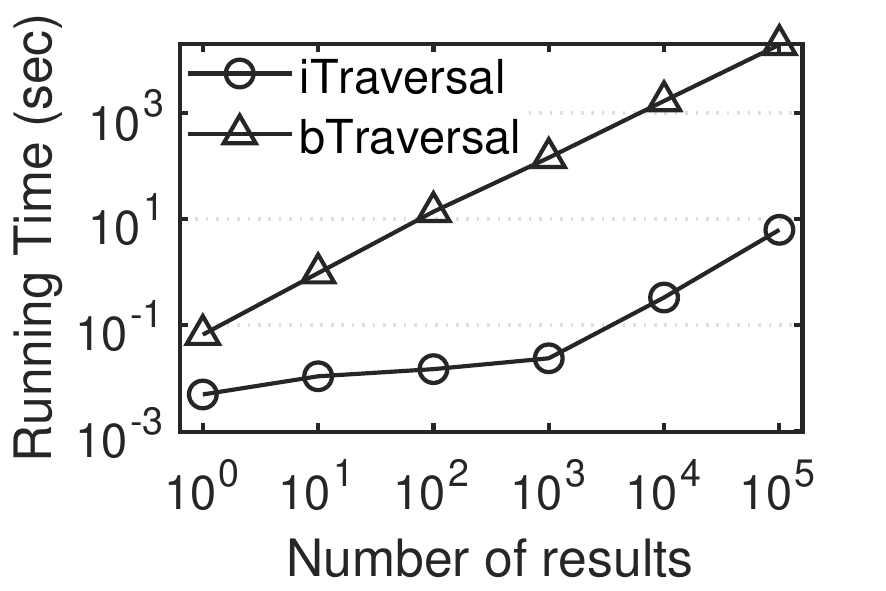}
		\end{minipage}
			&
	\begin{minipage}{3.50cm}
		\includegraphics[width=4.1cm]{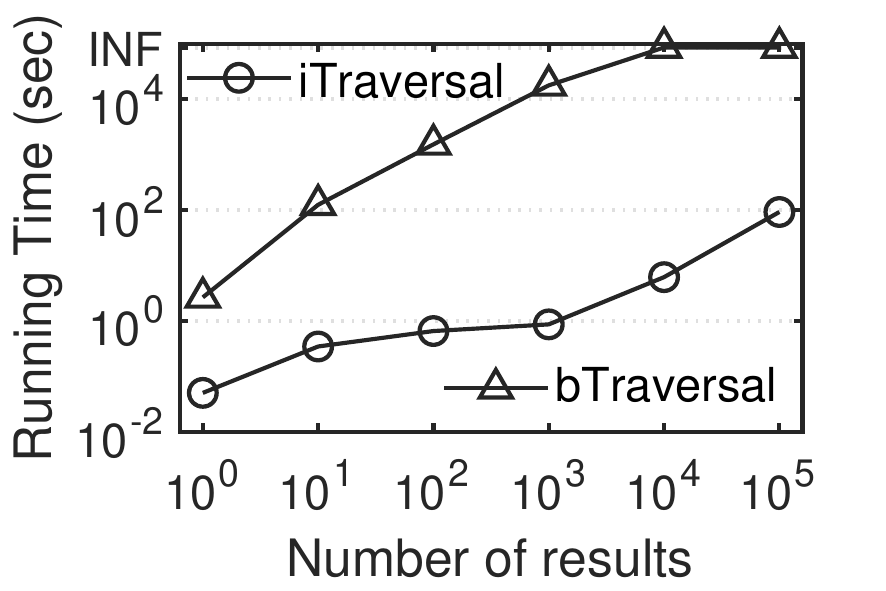}
	\end{minipage}
		\\
		(d) Varying \# of MBPs (Writer)
		&
		(e) Varying \# of MBPs (DBLP)
	\end{tabular}
	\caption{\roundD Results of running time (real datasets)}
	\label{fig:real}
\end{figure}

\smallskip\noindent{\textbf{Results of delay (real datasets).}}
The results are shown in Figure~\ref{fig:small_dataset}. Note that we use small datasets for this experiment since some baseline algorithms \texttt{iMB} cannot find all solutions on larger datasets {\roundE within INF} and thus the delay cannot be measured. In general, the delay of all algorithms increases with $k$ and \texttt{iTraversal} has the smallest delays. This is consistent with our theoretical analysis that the delay of \texttt{iTraversal} is polynomial while those of \texttt{iMB} and \texttt{FaPlexen} are exponential w.r.t. the size of the input bipartite graph. 
\begin{figure}[ht]
	\centering
	\centering
	\begin{tabular}{c c}
		
		\begin{minipage}{3.80cm}
			\includegraphics[width=4.1cm]{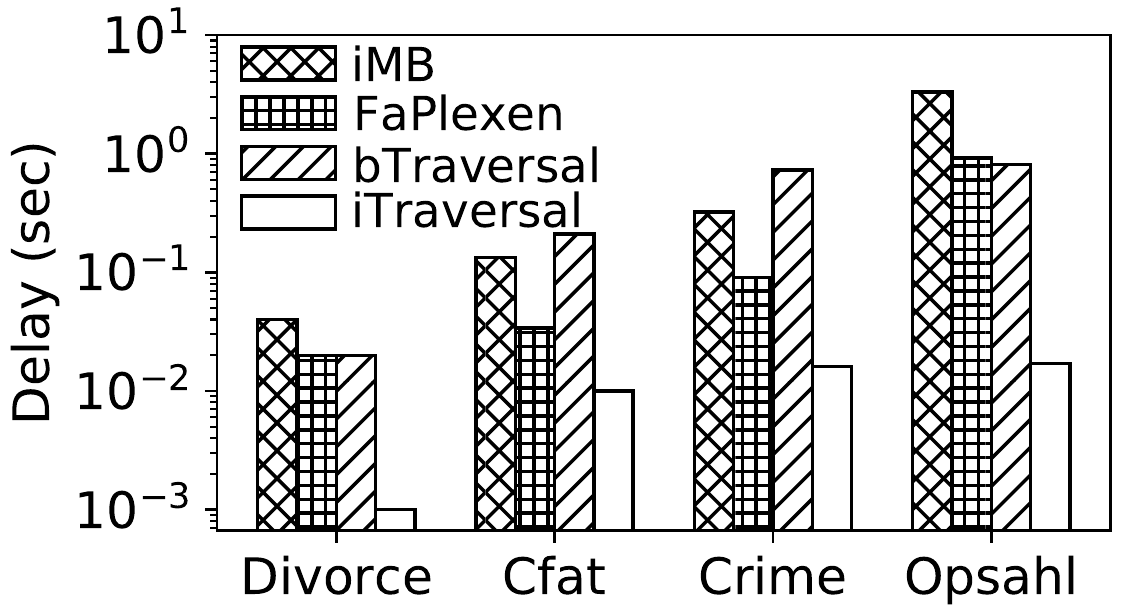}
		\end{minipage}
		&
		\begin{minipage}{3.80cm}
			\includegraphics[width=4.0cm]{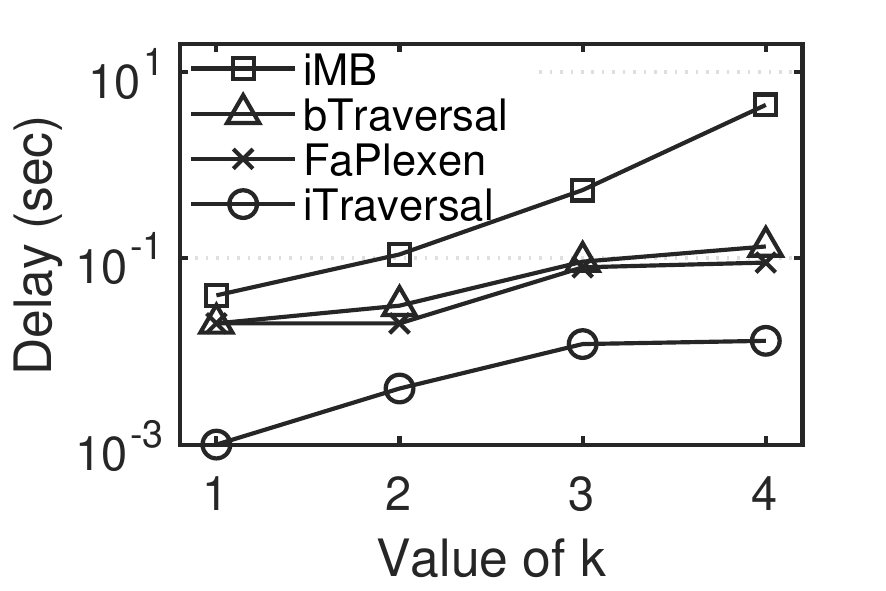}
		\end{minipage}		
		\\
		(a) Varying datasets ($k=1$)
		&
		(b) Varying $k$ (Divorce) 
	\end{tabular}
	\caption{\roundD Results of delay}
	\label{fig:small_dataset}
\end{figure}

\smallskip
\noindent{\textbf{Scalability test/Varying \# of vertices (synthetic datasets).}}
 The results when returning the first 1,000 MBPs are shown in Figure~\ref{fig:syn}(a).
\texttt{iTraversal} achieves at least 100$\times$ speedup compared with \texttt{bTraversal} and can handle large datasets with 100 million vertices and 1 billion edges {\roundE within INF} while \texttt{bTraversal} cannot.

\smallskip
\noindent{\textbf{Varying edge density (synthetic datasets).}}
 The results are shown in Figure~\ref{fig:syn}(b). \texttt{iTraversal} outperforms \texttt{bTraversal} by around 1-5 orders of magnitude. The speedup decreases as the graph becomes denser. Possible reasons include: (1) the density gap between the original almost-satisfying graph and the inflated graph narrows {\roundA (note that \texttt{bTraversal} involves a graph inflation step for implementing \texttt{EnumAlmostSat})} and {\roundB (2) the number of solutions decreases in dense bipartite graphs.}

\begin{figure}[]
	\centering
	\centering
	\begin{tabular}{c c}

		\begin{minipage}{3.50cm}
			\includegraphics[width=4cm]{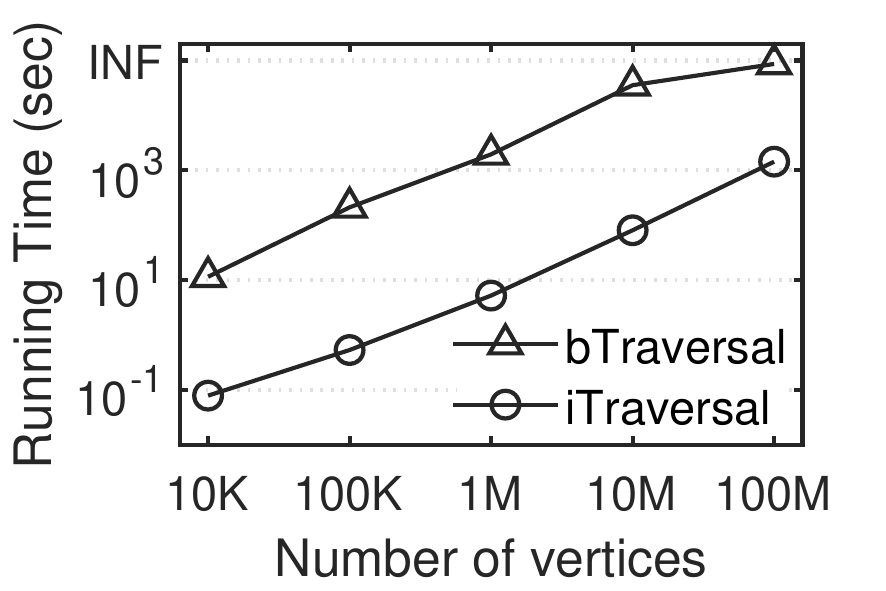}
		\end{minipage}
		&
		\begin{minipage}{3.50cm}
			\includegraphics[width=4cm]{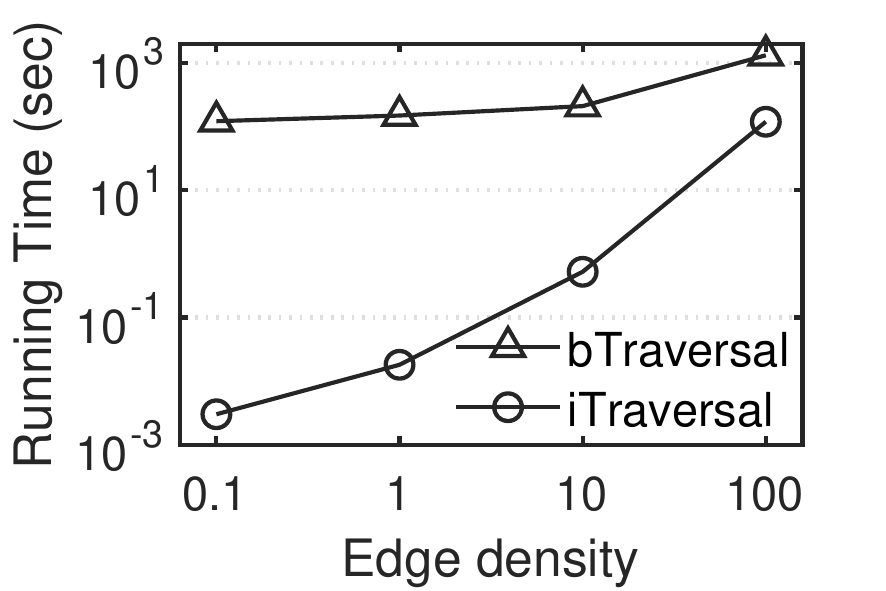}
		\end{minipage}
		\\
		(a) Varying \# of vertices
		&
		(b) Varying edge density
	\end{tabular}
	\caption{Results of running time (synthetic datasets)}
	\label{fig:syn}
\end{figure}

\smallskip
\noindent\textbf{Extension of \texttt{iTraversal} for enumerating large MBPs (real datasets)}. We compare the extension of \texttt{iTraversal} with \texttt{iMB}, for which some pruning techniques are developed for {\roundE enumerating large MBPs~\cite{yu2021efficient}}. We exclude \texttt{FaPlexen} and \texttt{bTraversal} for comparison because (1) \texttt{FaPlexen} cannot handle large datasets due to the graph inflation procedure and (2) \texttt{btraversal} cannot be extended with techniques for pruning small MBPs, but needs to enumerate all MBPs and then filter out those small ones. In addition, for both \texttt{iTraversal} and \texttt{iMB}, we perform a $(\theta-k)$-core decomposition~\cite{liu2019efficient} on the input bipartite graph first before enumerating large MBPs on them. This is due to the fact that each large MBP, i.e., a MBP with the size of both sides at least $\theta$, corresponds to a $(\theta-k)$-core, which is clear. The results of running time {\roundC on Writer and DBLP} are shown in Figure~\ref{fig:size_constraint} for varying $\theta$. We observe that the running time decreases as $\theta$ grows. {\roundB This is because both the size of $(\theta-k)$-core and the number of large MBPs decrease as $\theta$ grows.} \texttt{iTraversal} is faster than \texttt{iMB} by up to four orders of magnitude.

\begin{figure}[ht]
	\centering
	\centering
	\begin{tabular}{c c}
		\begin{minipage}{3.70cm}
			\includegraphics[width=4.1cm]{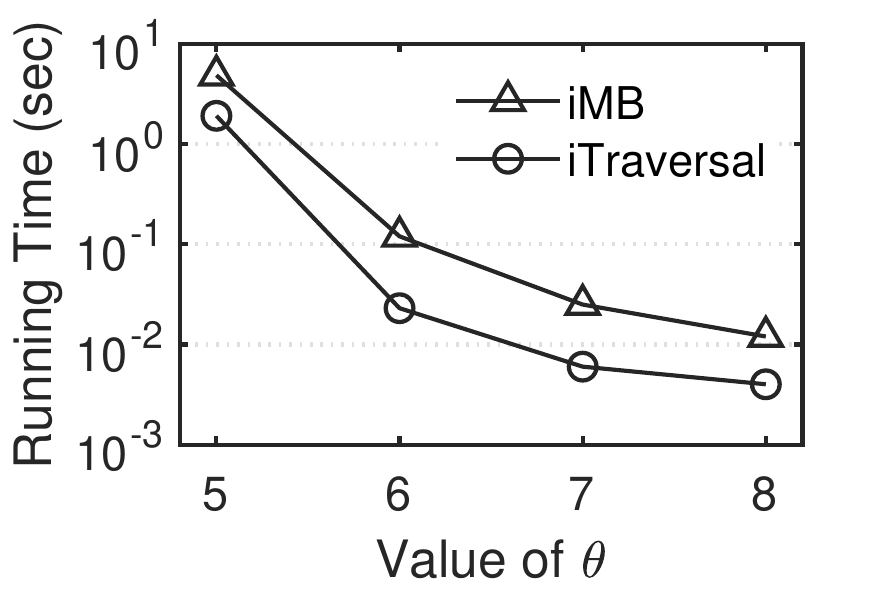}
		\end{minipage}
		&
		\begin{minipage}{3.70cm}
			\includegraphics[width=4.1cm]{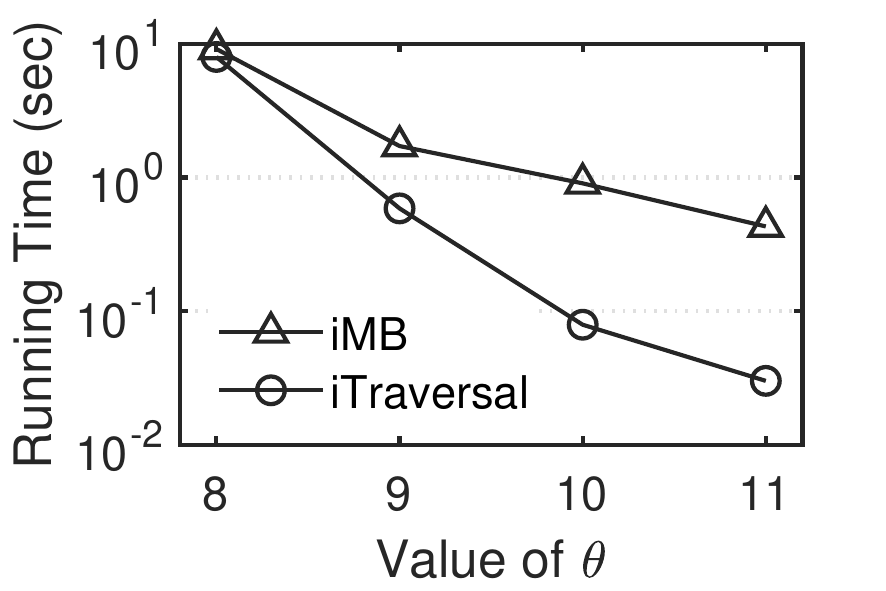}
		\end{minipage}	
		\\
		(a) Varying $\theta$ ($k$=1, Writer)
		&
		(b) Varying $\theta$ ($k$=1, DBLP)	
	\end{tabular}
	\caption{\roundD Results of enumerating large MBPs}
	\label{fig:size_constraint}
\end{figure}

\subsection{Performance Study of \texttt{iTraversal}}

\noindent{\textbf{\texttt{bTraversal} vs \texttt{iTraversal}}} 
We compare three different versions of \texttt{iTraversal}, namely (1) \texttt{iTraversal}: the full version involving left-anchored traversal, right-shrinking traversal and exclusion strategy, (2) \texttt{iTraversal-ES}: the full version without the exclusion strategy, and (3) \texttt{iTraversal-ES-RS}: \texttt{iTraversal-ES} without the right-shrinking traversal. We also consider the \texttt{bTraversal} for comparison. All these algorithms adopt the \texttt{L2.0+R2.0} algorithm for the \texttt{EnumAlmostSat} procedure for fair comparison. We measure the number of links of the solution graph and the running time. The results are shown in Figure \ref{fig:strategy}. Note that we use those small real datasets for this experiment only since \texttt{bTraversal} cannot find all solutions {\roundE within INF}. We set the maximum allowed number of links (UPP) as $10^{10}$. Figure \ref{fig:strategy}(a) and (b) show that  \texttt{iTraversal} would generate the sparsest solution graph and achieves up to 1000$\times$ speedup compared with \texttt{bTraversal}. These results also verify the effectiveness of left-anchored traversal, right-shrinking traversal, and exclusion strategy. Figure \ref{fig:strategy}(c) and (d) show the results of varying $k$. The scale of solution graph and the running time of all algorithms grow exponentially w.r.t. $k$. 
\begin{figure}[]
	\centering
	\centering
	\begin{tabular}{c c}
		\begin{minipage}{3.70cm}
			\includegraphics[width=4.0cm]{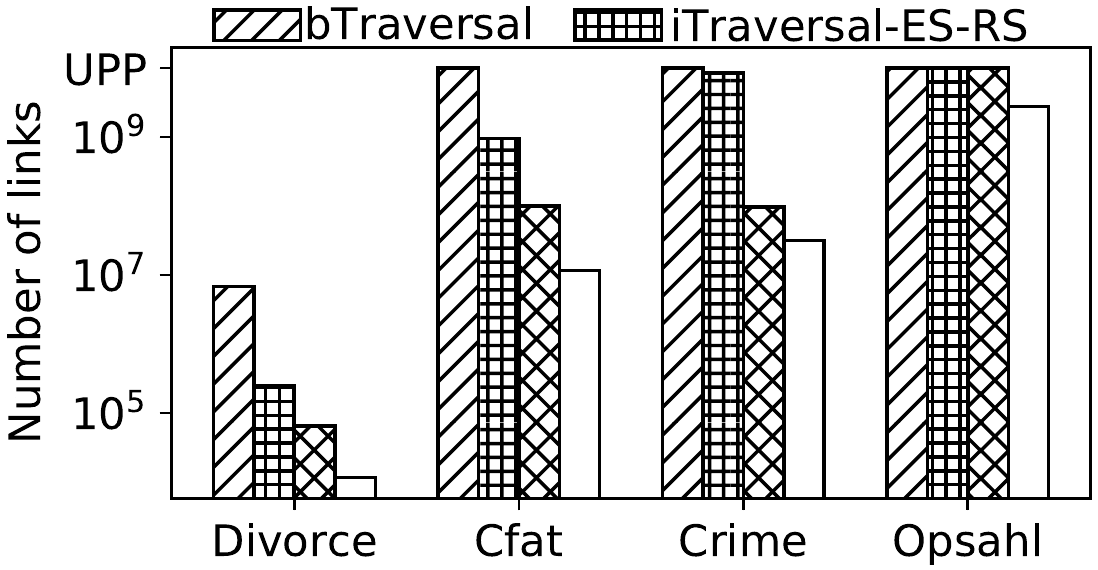}
		\end{minipage}
		&
		\begin{minipage}{3.70cm}
			\includegraphics[width=4.0cm]{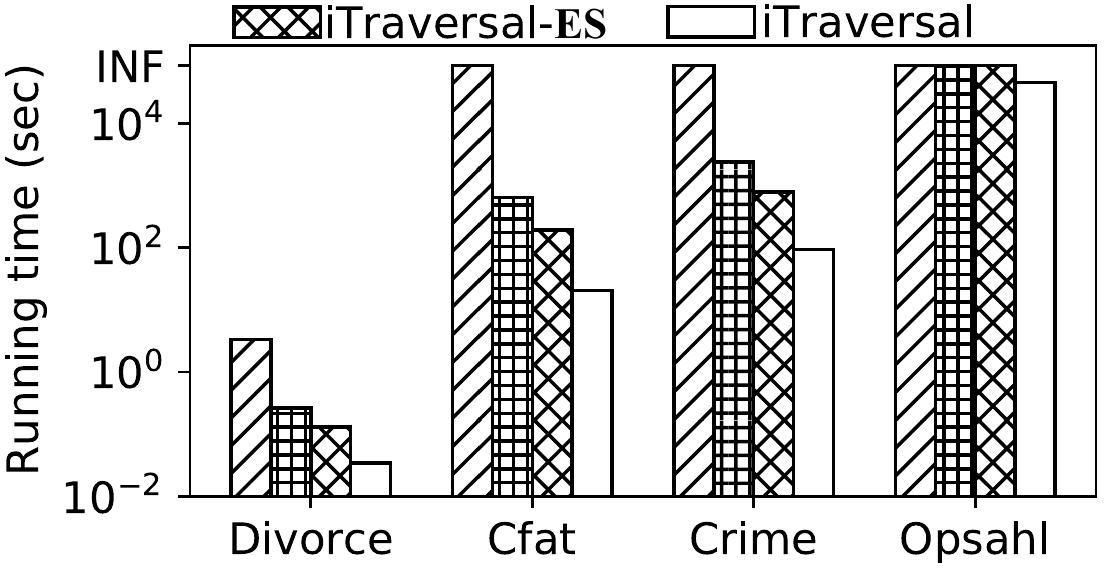}
		\end{minipage}	
		\\
		\ \ \ (a) Solution graph ($k=1$)
		&
		\ \ \ {(b) Running time ($k=1$)}	
		\\	
		\begin{minipage}{3.70cm}
			\includegraphics[width=4.1cm]{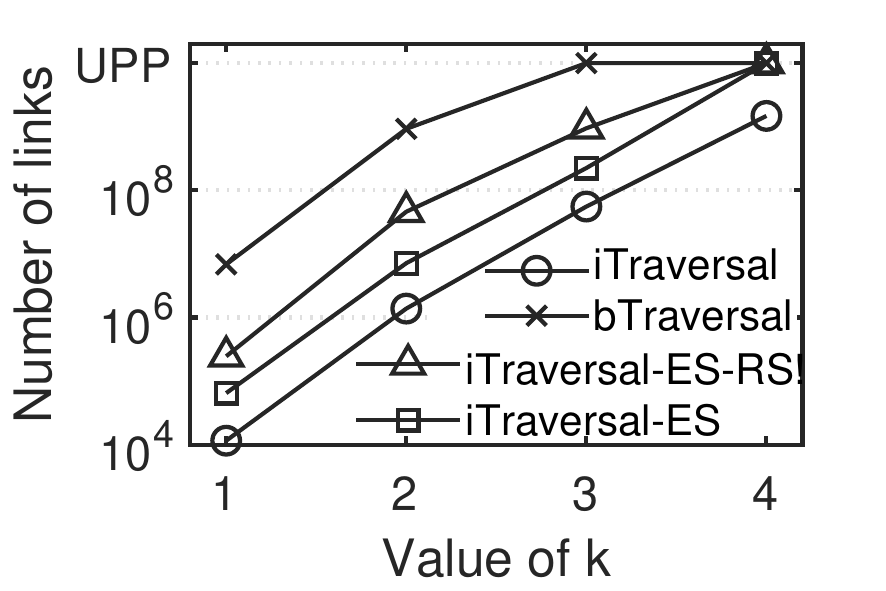}
		\end{minipage}
		&
		\begin{minipage}{3.70cm}
			\includegraphics[width=4.1cm]{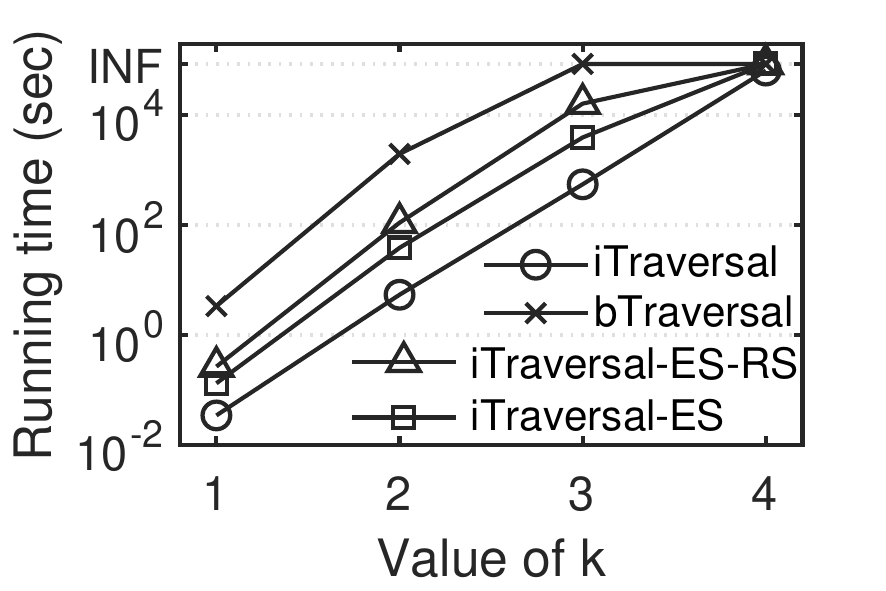}
		\end{minipage}	
		\\
		\ \ \ (c) Solution graph (Divorce)
		&
		\ \ \ (d) Running time (Divorce)	
	\end{tabular}
	\caption{\texttt{bTraversal} vs \texttt{iTraversal}}
	\label{fig:strategy}
\end{figure}

\smallskip
\noindent{\textbf{Performance study of \texttt{EnumAlmostSat}.}}
We consider different combinations of enumeration methods on $R$ and $L$ as described in Section~\ref{sec:enumalmostsat}. The algorithm \texttt{L1.0+R1.0} corresponds to the combination of ``refined enumeration on $L$: 1.0'' and ``refined enumeration on $R$: 1.0''. Similarly, we have \texttt{L1.0+R2.0}, \texttt{L2.0+R1.0}, and \texttt{L2.0+R2.0}. In addition, we consider a baseline method \texttt{Inflation} which conducts a graph inflation procedure and then uses an existing procedure for enumerating local maximal $k$-plexes \cite{DBLP:conf/sigmod/BerlowitzCK15}. 
We execute these algorithms 1,000 times with random almost-satisfying graphs that are constructed based on Writer {\roundB and DBLP}. {\roundA Specifically, we run $\texttt{iTraversal}$ on {\roundA a dataset (Writer or DBLP)}, collect the first 1,000 MBPs and then for each MBP $H=(L,R)$, we add to $H$ a random vertex $v\in \mathcal{L}\backslash L$.} We present the average running time in Figure \ref{fig:enum}. According to the results, the running time of all algorithms grows as $k$ increases and \texttt{L2.0+R2.0} is the fastest, achieving up to 1000$\times$ speedup compared with \texttt{Inflation}. 
\begin{figure}[ht]
	\centering
	\begin{tabular}{c c}
		\begin{minipage}{3.70cm}
			\includegraphics[width=4.0cm]{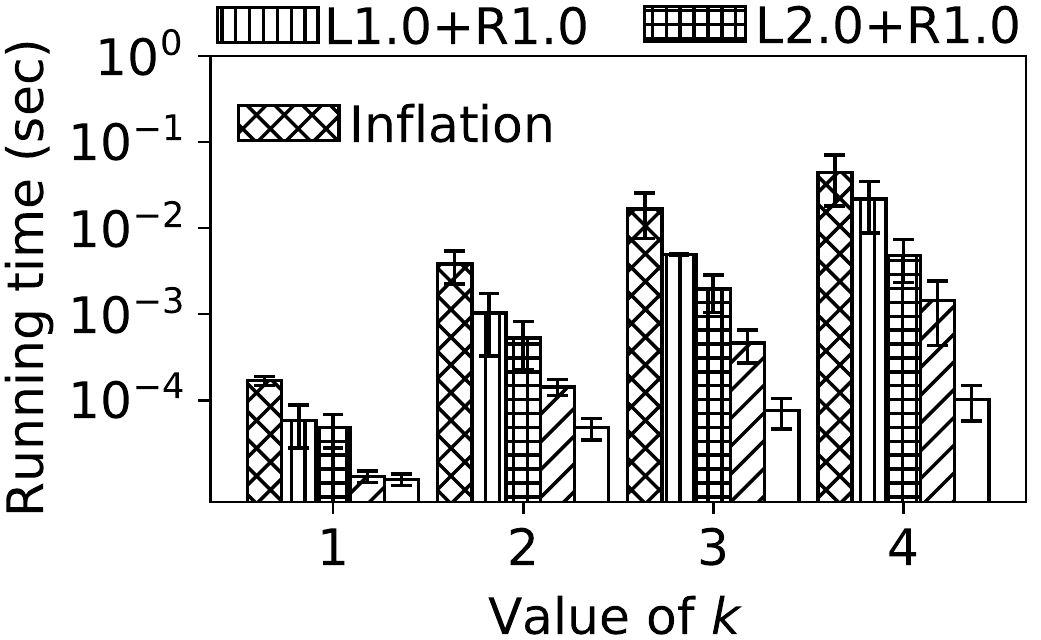}
		\end{minipage}
		&
		\begin{minipage}{3.70cm}
			\includegraphics[width=4.0cm]{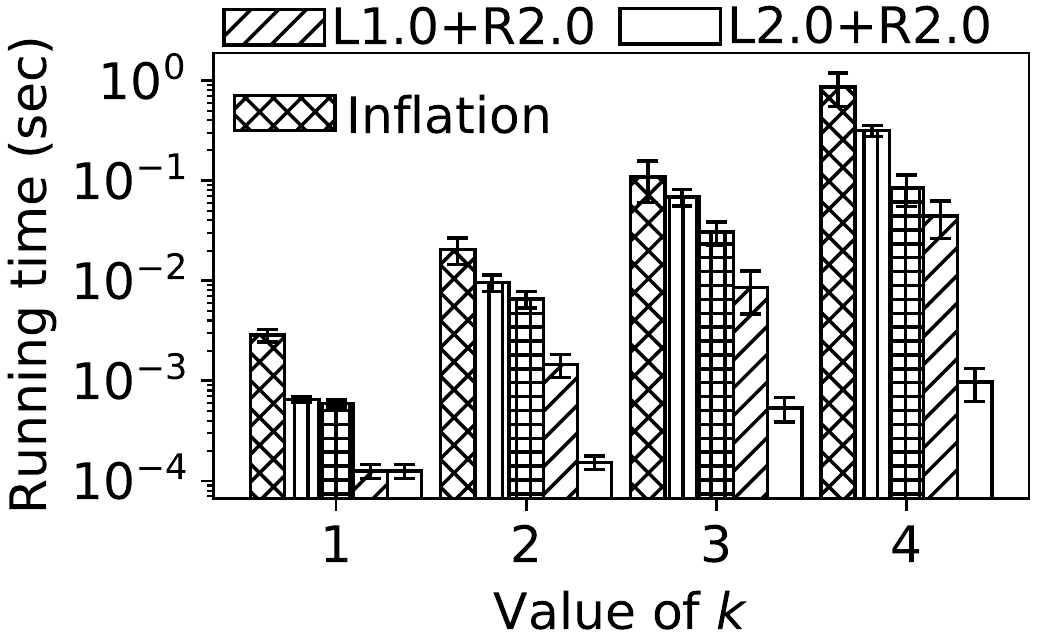}
		\end{minipage}		
		\\
		(a) Varying $k$ (Writer)
		&
		(b) Varying $k$ (DBLP)
	\end{tabular}
	\caption{Comparison among algorithms of \texttt{EnumAlmostSat}}
	\label{fig:enum}
\end{figure}

\smallskip
\noindent{\textbf{Left-anchored traversal vs Right-anchored traversal}.}
Recall that the left-anchored traversal is based on an initial solution $H_0 = (L_0, \mathcal{R})$. A symmetric option is to use an intial solution $H_0' = (\mathcal{L}, R_0)$ as described in Section~\ref{subsec:iTraversal}. We compare these two options by varying $k$ and measuring the running of time of returning the first 1,000 MBPs.
{\roundE The results are put in the technical report \cite{TR} for sake of the space.}
We observe from the results that these two options have similar results. {\roundA For left-anchored traversal, the scale of its solution graph depends on both sides, e.g., the number of almost-satisfying graphs is affected by $\mathcal{L}$ and the number of local MBPs within an almost-satisfying graphs is affected by both $\mathcal{L}$ and $\mathcal{R}$. There are no clear metrics to determine a dominating side.}


\subsection{Case Study: Fraud Detection}
\label{subsec:case-study}

{\ChengCommentB 
{\revision In this case study, we investigate four types of cohesive subgraphs, namely biclique, $k$-biplex, $(\alpha, \beta)$-core, and $\delta$-quasi-biclique (abbreviated as $\delta$-QB), for a fraud detection task in face of camouflage attacks~\cite{DBLP:conf/kdd/HooiSBSSF16}.}
%
Specifically, we use the ``Amazon Review Data'' (the software category)~\cite{dataset}, which contains 459,436 reviews on 21,663 softwares by 375,147 users. 
We then consider a scenario with the random camouflage attack~\cite{DBLP:conf/kdd/HooiSBSSF16},
which injects a fraud block with 2K fake users, 2K fake products (i.e., softwares), 200K fake comments, and 200k camouflage comments to the data. 
The fake comments (resp. camouflage comments) are randomly generated between pairs of fake users and fake products (resp. real products) such that each fake user has an equal number of fake comments and camouflage ones. 
The rationale behind the random camouflage attack is that in reality, fake users could be coordinated to comment on a set of products and/or deliberately post comments on some real products so as not to be spotted out~\cite{DBLP:conf/kdd/HooiSBSSF16}.
%
{\ChengComment {\revision We then find four types of cohesive subgraphs, namely biclique, $k$-biplex (with $k=1$ and $k=2$), $(\alpha,\beta)$-core, and $\delta$-QB (with $\delta=0.01$, 0.1, 0.2 and 0.3).
For biclique, $k$-biplex and $\delta$-QB, we explore different size constraints on the number of users (denoted by $\theta_L$) and the number of products (denoted by $\theta_R$).} For $(\alpha,\beta)$-core, we explore different settings of $\alpha$ and $\beta$ with $\alpha$=$\theta_R$ and $\beta$=$\theta_L$.}
{\roundD We classify all users and products that are involved in the found subgraphs as fake items and others as real ones.} We then measure the precision, recall and F1 score. 
{\roundD

The results are shown in Figure~\ref{fig:case_study1}, where $\theta_L (\beta)$ is fixed at 4 and $\theta_R (\alpha)$ varies from 3 to 7.}
We have the following observations.
%
%
\begin{itemize}[leftmargin=*]
\item \textbf{$k$-biplex.}
It has both high precision and high recall with appropriate settings of $\theta_R$ (e.g., for $1$-biplex, it has precision around $90\%$ and recall around $100\%$ when $\theta_R = 5$). 
This shows that in a scenario with random camouflage attacks, fake users and products can be effectively spotted out by finding $k$-biplexes. In this scenario, a group of fake users (e.g., those who are coordinated) would connect most of products in a pool (e.g., those that are to be promoted by paying the fake users) and deliberately disconnect few of the products, and this phenomenon is effectively captured by $k$-biplexes. {\revision In addition, we find that $k$-biplexes with a large $k$ would involve many disconnections and become less cohesive. Hence, small $k$'s are usually used for $k$-biplex in real applications, e.g., the setting $k=1$ achieves the best result for this fraud detection application.}
\item \textbf{biclique.}
 {\revision It 
 usually has very low recall as shown in Figure~\ref{fig:case_study1}(b)} (e.g., it has recall around 55\% when $\theta_R = 4$ and close to 0 when $\theta_R \ge 5$), since biclique requires complete connections which is too strict. {\roundD We remark that if no users are found, the precision and F1 score are undefined, marked as ``ND'' in Figure~\ref{fig:case_study1}(a) and (c).}
\item \textbf{$(\alpha,\beta)$-core.}
It has high recall but very low precision, since it would usually find large but sparse subgraphs (e.g., it has precision below 25\% in all settings shown in the figure and not beyond 30\% even with larger $\alpha$'s). {\revision 
{\ChengCommentC We note that all ($\alpha, \beta$)-cores found are connected.}
In addition, we observe that its precision grows as both $\alpha$ and $\beta$ become larger, but then the recall would drop. By exploring various settings of $\alpha$ and $\beta$, we find that when $\alpha=30$ and $\beta=12$, it has the best F1 score of 0.64, which is still worse than that of 1-biplex (i.e., 0.96).}

{\revision \item \textbf{$\delta$-QB.}  With larger $\delta$'s, it would allow more disconnections in a $\delta$-QB, and thus the precision decreases and the recall increases. In addition, with larger $\theta_R$'s, it would find fewer $\delta$-QBs that are large, and thus the precision increases and the recall decreases. By trying various settings, we found that when $\delta=0.2$ and $\theta_R=6$, it has the best F1 score of 0.88, which is worse than that of 1-biplex. We also note that a $k$-biplex with the sizes of both sides at least $\theta$ is a $k/\theta$-QB, and a $\delta$-QB with the sizes of both sides at most $\theta$ is a $\lceil \theta\delta\rceil$-biplex. With typical settings of $k$ (e.g., below 5) and $\delta$ (e.g., a small real number below 1), $k$-biplexes are usually denser than $\delta$-QBs since for the latter, there could be possibly be many disconnections when the sizes are large.
}
\end{itemize}
In summary, $k$-biplexes are superior over other structures for fraud detection, as confirmed by the F1 scores shown in Figure~\ref{fig:case_study1}(c) (with the best of each method in bold).}
%
%

\begin{figure}[]
	\centering
	\centering
	\begin{tabular}{c c}
		\begin{minipage}{3.70cm}
			\includegraphics[width=4.0cm]{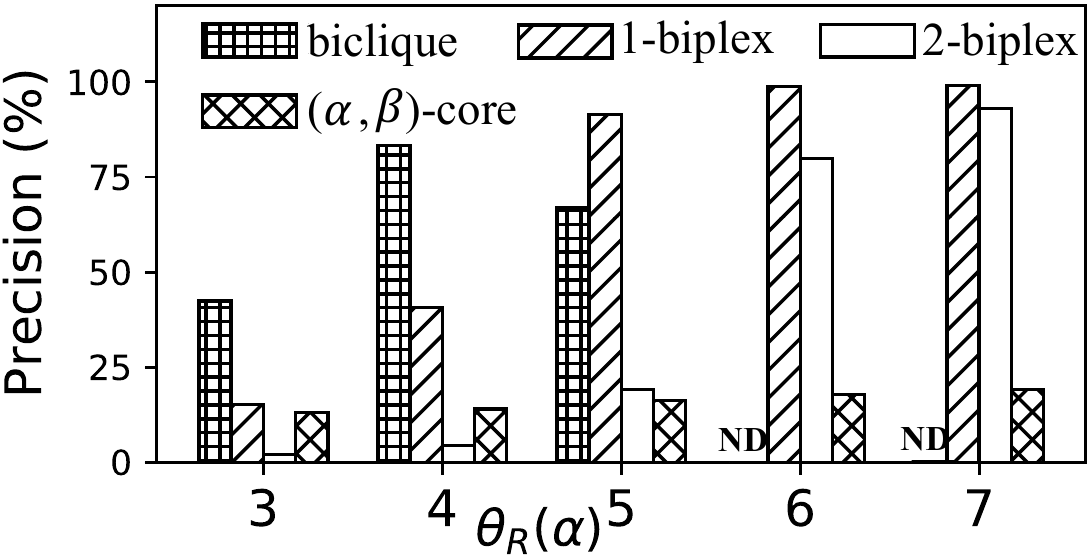}
		\end{minipage}
		&
		\begin{minipage}{3.70cm}
			\includegraphics[width=4.0cm]{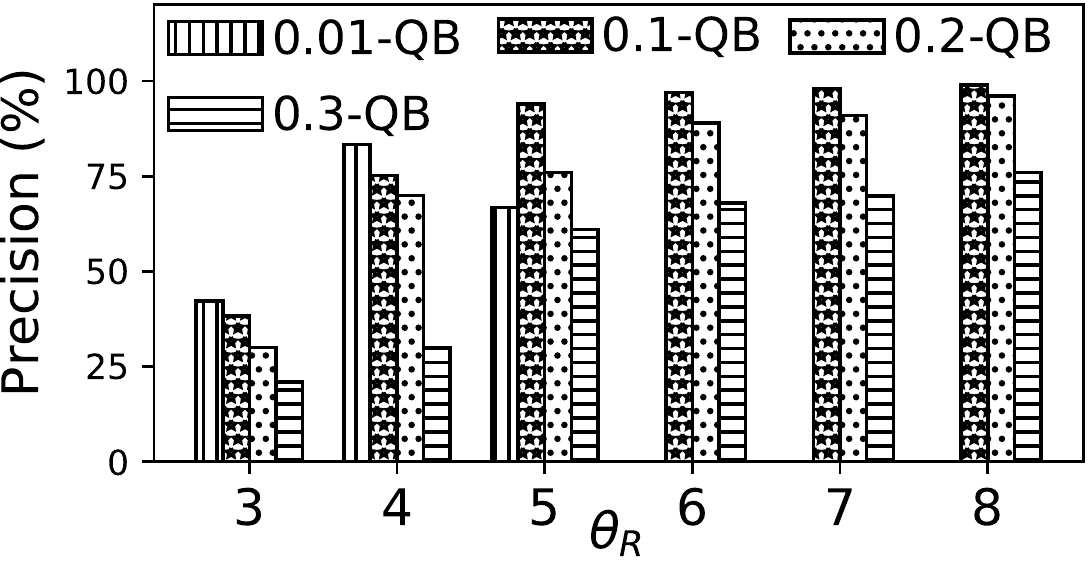}
		\end{minipage}
	\end{tabular}
	\\
	(a) Precision
	\vspace{0.10in}
	\\
	\begin{tabular}{c c}
		\begin{minipage}{3.70cm}
			\includegraphics[width=4.0cm]{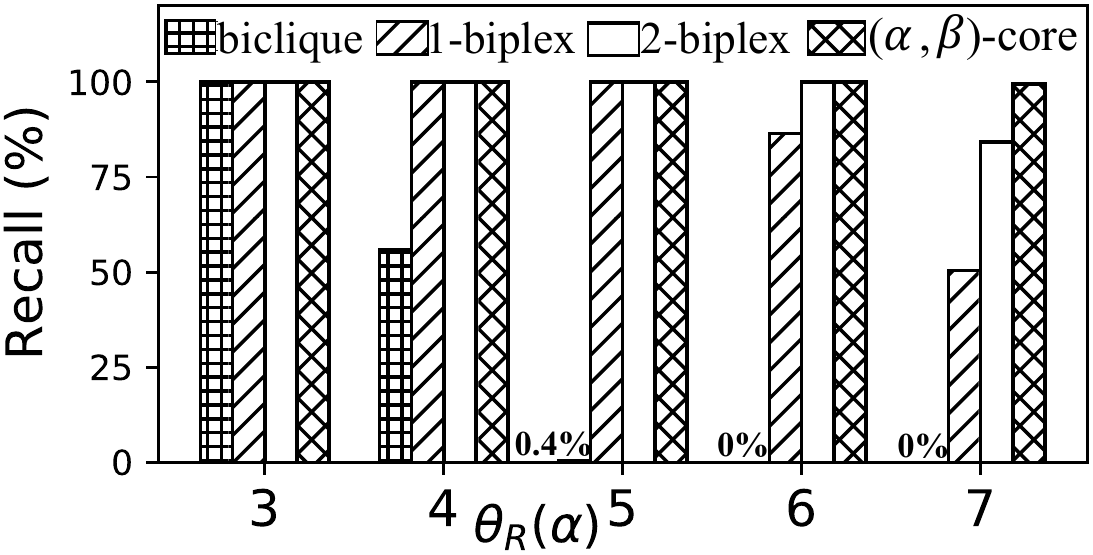}
		\end{minipage}
		&
		\begin{minipage}{3.70cm}
			\includegraphics[width=4.0cm]{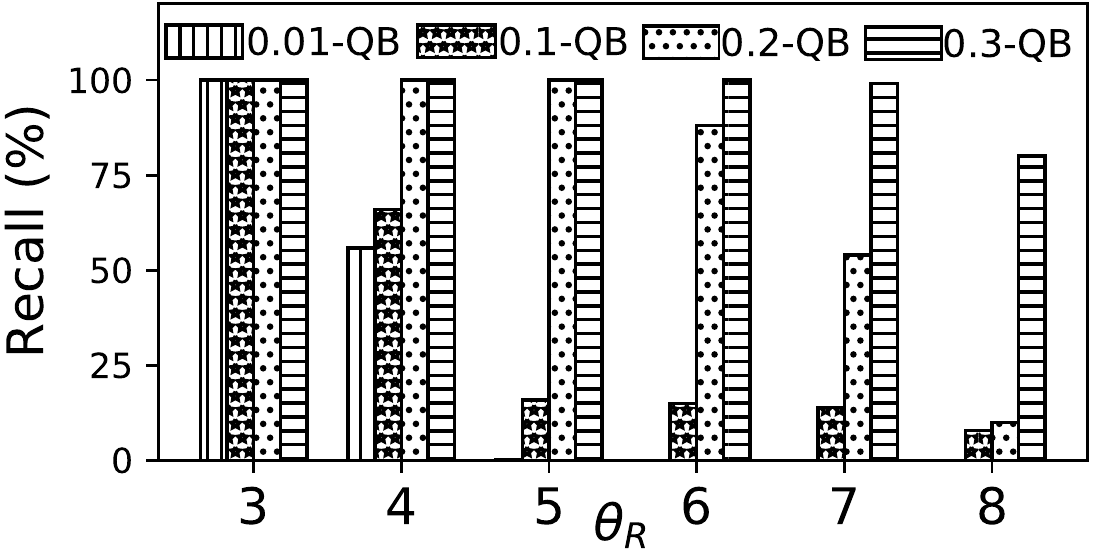}
		\end{minipage}		
	\end{tabular}
	\\
	(b) Recall
	\vspace{0.10in}
	\\
	\begin{minipage}{8cm}
			\includegraphics[width=8cm]{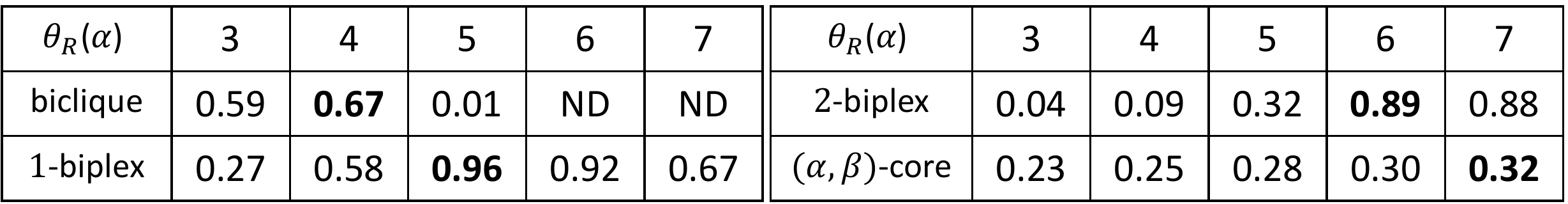}
	\end{minipage}
	\begin{minipage}{8cm}
			\includegraphics[width=8cm]{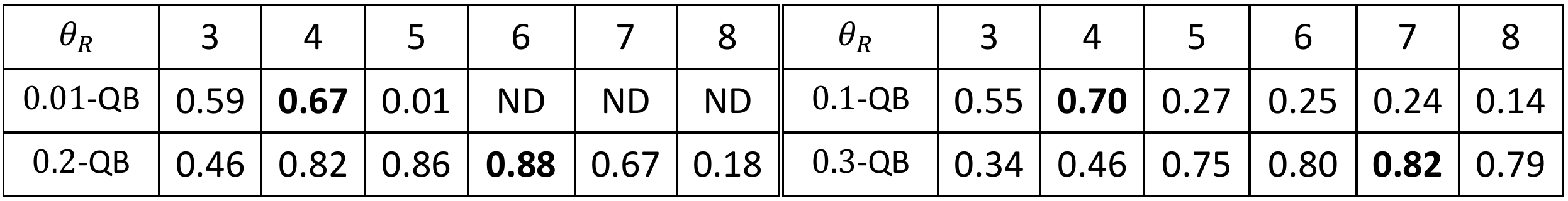}
	\end{minipage}
	\\
	(c) F1 score 
	
	\caption{\revision Case Study: Fraud Detection}
	\label{fig:case_study1}
\end{figure}

\balance

\section{Related Work}
\label{related}

{\ChengCommentB
\noindent{\textbf{Dense subgraphs of bipartite graphs.}} 
Extensive studies have been conducted on finding and/or enumerating dense subgraphs of bipartite graphs, including bicliques \cite{zhang2014finding,DBLP:conf/ijcai/AbidiZCL20,mukherjee2016enumerating,DBLP:journals/pvldb/LyuQLZQZ20,DBLP:journals/corr/abs-2007-08836,kloster2019mining}, $(\alpha,\beta)$-cores \cite{liu2019efficient}, $k$-biplexes~\cite{DBLP:journals/sadm/SimLGL09,yu2021efficient}, $k$-bitrusses \cite{DBLP:conf/dasfaa/Zou16,wang2020efficient} and quasi-bicliques \cite{DBLP:conf/cocoon/LiuLW08,wang2013near,ignatov2018mixed,DBLP:conf/icfca/Ignatov19}. 
%
Existing studies of bicliques focus on maximal biclique enumeration~\cite{zhang2014finding,DBLP:conf/ijcai/AbidiZCL20,mukherjee2016enumerating,kloster2019mining} and maximum biclique discovery~\cite{DBLP:journals/pvldb/LyuQLZQZ20,DBLP:journals/corr/abs-2007-08836}.
%
$(\alpha,\beta)$-cores have been used in applications such as recommendation systems~\cite{ding2017efficient} and community search~\cite{DBLP:conf/ssdbm/HaoZW020,DBLP:journals/corr/abs-2011-08399}.
Existing studies of $k$-biplexes focus on large maximal $k$-biplex enumeration~\cite{DBLP:journals/sadm/SimLGL09,yu2021efficient}.
A $k$-bitruss~\cite{DBLP:conf/dasfaa/Zou16,wang2020efficient} is a bipartite graph where each edge is contained in at least $k$ butterflies,
where a butterfly corresponds to a complete 2$\times$2 biclique~\cite{DBLP:journals/pvldb/WangLQZZ19}.
%
There are two types of definition for a quasi-biclique $H=(L,R)$, namely (1) $\delta$-quasi-biclique~\cite{DBLP:conf/cocoon/LiuLW08}, where each vertex in $L$ (resp. $R$) misses at most $\delta \cdot |R|$ (resp. $\delta \cdot |L|$) edges with $\delta\in [0,1)$) and (2) $\gamma$-quasi-biclique~\cite{ignatov2018mixed}, where at most $\gamma \cdot |L|\cdot|R|$ edges can be missed with $\gamma\!\in\! [0,1)$.
%
%
Existing studies of quasi-bicliques focus on finding the maximum quasi-biclique \cite{DBLP:conf/cocoon/LiuLW08,wang2013near,ignatov2018mixed,DBLP:conf/icfca/Ignatov19}. {\roundA There are some other studies which find subgraphs with a certain density and degree~\cite{DBLP:conf/kdd/MitzenmacherPPT15,DBLP:conf/bigdata/LiuSC13}.}
In this paper, we focus on $k$-biplex for reasons as discussed in Section~\ref{sec:introduction} (i.e., $k$-biplex imposes strict enough requirements on connections within a subgraph, tolerates some disconnections, and satisfies the hereditary property).}


\smallskip
\noindent{\textbf{Maximal subgraphs enumeration.}}
In general, there are two types of methodology, namely Bron-Kerbosch scheme~\cite{bron1973algorithm} and reverse search~\cite{DBLP:journals/dam/AvisF96,DBLP:journals/jcss/CohenKS08},
for enumerating maximal subgraphs of desired properties. 
%
%
The original Bron-Kerbosch (BK) algorithm is proposed in \cite{bron1973algorithm} for enumerating maximal cliques. 
BK uses branch-and-bound and backtracking to filter out the branches that cannot yield desired subgraphs.
Quite a few variants of BK have been proposed for solving various enumeration problems, including (bi)cliques~\cite{DBLP:journals/tcs/TomitaTT06,DBLP:journals/jea/EppsteinLS13,zhang2014finding,DBLP:conf/ijcai/AbidiZCL20},  $k$-plexes~\cite{DBLP:conf/aaai/ZhouXGXJ20,DBLP:conf/kdd/ConteMSGMV18}, signed cliques~\cite{chen2020efficient,DBLP:journals/tkde/LiDQWXYQ21}, motif cliques~\cite{DBLP:conf/icde/HuCCSFL19}, temporal cliques~\cite{DBLP:journals/snam/HimmelMNS17}, quasi-bicliques~\cite{DBLP:journals/sadm/SimLGL09} and etc. 
{\roundA However, the time complexity of BK and its variants is polynomial w.r.t. the worst-case number of desired subgraphs.}
%
%
In addition, these algorithms output desired subgraphs with exponential delay.
%
Reverse search was first proposed as a general framework for enumerating subgraphs~\cite{DBLP:journals/dam/AvisF96}. 
It defines a ``successor'' function to traverse from one solution to the others, which corresponds to a DFS on an implicit solution graph where the nodes represent desired subgraphs and the links capture the traversals from desired subgraphs to other desired subgraphs.  
{\roundB 
Recently, this framework has been adopted to solve various enumeration problems, e.g., independent sets \cite{DBLP:journals/siamcomp/TsukiyamaIAS77,DBLP:conf/spire/ConteGMUV17}, cliques \cite{DBLP:conf/swat/MakinoU04,DBLP:journals/algorithmica/ChangYQ13,DBLP:journals/algorithmica/ConteGMV20}, {\ChengComment $k$-plexes~\cite{DBLP:conf/sigmod/BerlowitzCK15}, and general structures that satisfy the hereditary property \cite{DBLP:conf/stoc/ConteU19, DBLP:journals/jcss/CohenKS08,cao2020enumerating}}}.
The algorithms following this framework achieve an output-sensitive time complexity that is proportional to the number of desired subgraphs (within the input graph). Moreover, there is a polynomial time guarantee on the time cost per solution. 
In this paper, we adopt the reverse search framework and develop various techniques for improving and instantiating the framework for enumerating maximal $k$-biplexes.
{\ChengCommentB We note that our techniques of improving the framework for $k$-biplexes are new and different from those that were developed for other specific structures such as independent sets and cliques and the latter cannot be applied to $k$-biplexes.}


\section{Conclusion}
\label{sec:conclusion}
In this paper, we study the \textit{maximal} $k$\textit{-biplex enumeration} problem. 
We develop an efficient reverse search-based algorithm with the polynomial delay guarantee. 
Extensive experiments on real and synthetic datasets demonstrate that our algorithm outperforms the existing methods in terms of total running time and delay significantly. 
{\roundA
In the future, we will investigate efficient parallel and distributed implementations. Another interesting research direction is to adapt the proposed reverse search-based algorithm to enumerate some other cohesive subgraphs over bipartite graphs.}

\section{Acknowledgement}
The research of Cheng Long and Kaiqiang Yu is supported by the Ministry of Education, Singapore, under its Academic Research Fund (Tier 2 Award MOE-T2EP20220-0011 and Tier 1 Award RG20/19 (S)) and by the Nanyang Technological University Start-Up Grant from the College of Engineering under Grant M4082302. Any opinions, findings and conclusions or recommendations expressed in this material are those of the author(s) and do not reflect the views of the Ministry of Education, Singapore. The research of Shengxin Liu is supported by the National Natural Science Foundation of China under Grant 62102117. The research of Da Yan is supported by NSF OAC-1755464 (CRII).

\balance
\clearpage
\balance
\bibliographystyle{ACM-Reference-Format}
\bibliography{SIGMOD_biplex}


\end{document}